\documentclass[12pt]{article}
\usepackage{amsfonts,amsmath}
\usepackage[mathscr]{eucal}
\usepackage{amssymb}
\usepackage{bm}
\usepackage{bbold}
\usepackage{amsthm}
\usepackage{multirow}
\theoremstyle{plain}
\newtheorem{thm}{Theorem}
\newtheorem{defi}{Definition}

\newtheorem{lem}{Lemma}

\def\theequation{\arabic{section}.\arabic{equation}}

\newcommand{\be}{\begin{eqnarray}}
\newcommand{\ee}{\end{eqnarray}}
\newcommand{\nn}{\nonumber \\}
\newcommand{\lb}{\label}
\newcommand{\p}[1]{(\ref{#1})}
\renewcommand{\u}{\underline}

\begin{document}

\begin{titlepage}

\vspace*{0.2cm}

\renewcommand{\thefootnote}{\star}
\begin{center}

{\LARGE\bf  Generic HKT geometries in the harmonic superspace approach  }\\

\vspace{0.5cm}

\vspace{1.5cm}
\renewcommand{\thefootnote}{$\star$}

{\quad \large\bf Sergey~Fedoruk} ${}^\star$,
\quad {\large\bf Evgeny~Ivanov} ${}^\star$,
\quad {\large\bf Andrei~Smilga} ${}^\ast$
 \vspace{0.5cm}

{${}^\star$ \it Bogoliubov Laboratory of Theoretical Physics, JINR,}\\
{\it 141980 Dubna, Moscow region, Russia}  \\
\vspace{0.1cm}

{\tt fedoruk@theor.jinr.ru, eivanov@theor.jinr.ru}\\
\vspace{0.7cm}

{${}^\ast$\it SUBATECH, Universit\'e de Nantes,}\\
{\it 4 rue Alfred Kastler, BP 20722, Nantes 44307, France}\\
\vspace{0.1cm}

{\tt smilga@subatech.in2p3.fr}\\

\end{center}
\vspace{0.2cm} \vskip 0.6truecm \nopagebreak

\begin{abstract}
\noindent We explain how a generic HKT geometry
can be derived using the language of ${\cal N} = 4$ supersymmetric
quantum mechanics. To this end, one should consider a Lagrangian
involving several ({\bf 4}, {\bf 4}, {\bf 0}) multiplets defined in
harmonic superspace and subject to nontrivial harmonic constraints.
Conjecturally, this general construction worked out in \cite{DI}
gives  a complete classification of all HKT geometries. Each such
geometry is generated by two different functions (potentials) of a special
type that depend on harmonic superfields and on harmonics.

  Given these two potentials, one can derive the vielbeins,  metric, connections and curvatures,
  but this is not so simple: one should solve
 rather complicated differential equations.
 We  illustrate the general construction by giving a detailed derivation of the metric for the hyper-K\"ahler
 Taub-NUT manifold. In the generic case, we arrive at an HKT geometry. In this paper, we give a simple proof
 of this assertion.

\end{abstract}

\vspace{1cm}

\phantom{Keywords: }

\newpage

\end{titlepage}

\setcounter{footnote}{0}

\setcounter{equation}0
\section{Introduction}

The language of supersymmetric quantum mechanics (SQM) is most adequate
and convenient for studying problems of differential geometry. This was
understood back in the eighties when E.~Witten showed how supersymmetry allows one to describe
in very simple terms the classical de Rham complex \cite{Witten}
and L.~Alvarez-Gaum\'e, D.~Friedan and P.~Windey gave a simple supersymmetric
proof of the Atiyah-Singer theorem
\cite{ASsusy}.~\footnote{This was done almost for all cases except for the index of the Dolbeault
operator for non-K\"ahler complex manifolds. This ``dark corner'' was recently illuminated in \cite{HRR}.}

The SQM methods not only allowed one to reproduce what was already known, but
 also to derive
many new mathematical results. This particular paper is devoted to the classification of the HKT
geometries.~\footnote{The abbreviation ``HKT'' means ``hyper-K\"ahler with torsion''. This term is somewhat misleading because
the HKT manifolds are not hyper-K\"ahler and not even K\"ahler, but now
it is firmly established in the literature, and we will use it. The HKT
geometries were discovered by physicists as geometries associated with supersymmetric
sigma models of a special kind \cite{GHR,HoPa1,strongHKT-def,GPS} and {\it then} attracted
a considerable attention of mathematicians (see e.g.
\cite{Grant,Verb}). It is worth noting that the first example of models with HKT geometry (as it was realized later) was ${\cal N}=(4, 4)$
supersymmetric extension of $2D$ WZNW $SU(2)$ sigma model \cite{IKr}, \cite{SSTVP}.}
This problem was solved in \cite{DI}. To do this, one should consider a supersymmetric sigma model endowed with the
extended ${\cal N} = 4$ supersymmetry (and having thus two pairs of complex supercharges). The corresponding superfield
Lagrangian depends on several ({\bf 4}, {\bf 4}, {\bf 0}) multiplets.~\footnote{We use the notation suggested in \cite{PT}.
 The first number counts the bosonic dynamical variables, the second number counts the fermions dynamical variables and the
third number the bosonic auxiliary fields.}

To define a ({\bf 4}, {\bf 4}, {\bf 0}) multiplet, we introduce as a starting point an extended ${\cal N} = 4\,,$ $d=1$
superspace
$(t; \theta_i, \bar \theta^i)$, $i=1,2$. Then we ``harmonize'' it \cite{harm1,harm}~\footnote{The technique of harmonic superspace
is explained in the monograph
\cite{harm} and, for the one-dimensional case (where superspace becomes
``supertime''), in Ref. \cite{IvLe}. In our presentation we basically follow the latter paper.} by defining $\theta^\pm = u_i^\pm \theta^i$, where $u_i^+$ is a complex
unitary spinor parameterizing the automorphism group $SU(2)$ of the extended superalgebra, and  $u_i^-$
is its complex conjugate.  A generic superfield depends on $(t; \, \theta^+,  \bar \theta^+, \theta^-, \bar\theta^-)$ and on the
harmonics $u_i^\pm$, but a special role is played by Grassmann-analytic superfields that depend, besides the harmonics, only on $\theta^+, \bar\theta^+$ and
the ``analytic time'' \p{tA}. A G-analytic superfield $\Phi$ satisfies the constraints $D^+ \Phi = \bar D^+ \Phi = 0$, which are quite
analogous to the constraints defining the chiral superfields. A G-analytic superfield is characterized by its harmonic charge --- an integer
eigenvalue of the operator \p{D0}.

Consider a G-analytic superfield $q^{+a}$ carrying unit harmonic charge and an extra doublet index $a$. To describe a {\it linear}
({\bf 4}, {\bf 4}, {\bf 0}) multiplet, we impose the additional harmonic constraint
\be
\lb{D++q=0}
D^{++} q^{+a} \ =\ 0 \, ,
\ee
where the harmonic derivative  $D^{++}$ is defined in \p{D++}, and
also require that
\be
\lb{cons-tilde}
\widetilde{q^+_a}  = \varepsilon^{ab} q^+_b   \equiv q^{+a}\, ,
\ee
where  the ``tilde'' conjugation is the superposition of the standard complex conjugation and the antipodal transformation of the harmonics, $\widetilde {u^{\pm}_j} = \varepsilon^{jk} u^{\pm}_k$ \cite{harm}.
Then the component expansion of  $q^{+a}$ reads~\footnote{Here $t$ is actually
the analytic time $t_A$, but we do not display the index $A$ anymore.}
\be
\lb{q+lin}
q^{+a} \ =\ x^{ja} (t) u_j^+ + \theta^+ \chi^a(t) + \bar \theta^+ \bar \chi^a(t) - 2i\,\theta^+ \bar \theta^+ {\dot x}^{ja} u_j^- \, ,
\ee
where the bosonic field $ x^{ja}$ is pseudoreal,
\be
\lb{pseudoreal}
(x^{ja})^* \ =\  \varepsilon_{jk} \varepsilon_{ab} x^{kb} \equiv x_{ja}\, ,
\ee
and
\be
\lb{conj-chia}
(\chi^a)^* \ =\ \bar\chi_a = \varepsilon_{ab} \bar \chi^b
 \ee
 ($\varepsilon_{12} = 1$).
The superfield $q^{+a}$   includes four real bosonic and four real fermionic component fields.

Now we  take $n$ such multiplets or, which is the same, assume that the index $a$ in
\p{q+lin} runs from 1 to $2n$. The constraints \p{cons-tilde} can then be
rewritten as
\be
\lb{cons-tilde-Omega}
\widetilde{q^+_a}  = \Omega^{ab} q^+_b  \equiv q^{+a} \, ,
\ee
where
\be
\lb{Omega}
\Omega^{ab} \ =\ - \Omega_{ab} \ =\ -{\rm diag} \left( i\sigma_2, \ldots, i\sigma_2 \right)
\ee
defines an antisymmetric symplectic form, $\Omega^{ab} \Omega_{bc} = \delta^a_c$.
A general supersymmetric action reads
\be
\lb{S-lin}
S \ =\ -\frac18 \int dt du d^4\theta \,  {\cal L} (q^{+a}, q^{-b}, u^\pm) \, ,
\ee
where
\be
\lb{q-def}
q^{-a} = D^{--} q^{+a}
\ee
with $D^{--}$ defined in \p{D--} and the numerical normalization factor is chosen to match the notations  of \cite{DI}.
If expanding this Lagrangian into components, one derives
\be
L \ =\ \frac 12 \,g_{ja, kb} \, {\dot x}^{ja}  {\dot x}^{kb} + \ {\rm fermion\ terms} ,
\ee
where the metric $g_{ja, kb}$ is expressed via  double derivatives of ${\cal L}$.

The system of an arbitrary number of linear
({\bf 4}, {\bf 4}, {\bf 0}) multiplets was studied in a different approach in \cite{FS-HKT}. It was shown that a $4n$-dimensional manifold
with the metric $g_{ja, kb}$ is a HKT manifold.
However, such system describes only a rather limited class of
HKT-manifolds: namely, the manifolds with vanishing {\it Obata curvature}
\cite{Obata}.~\footnote{We will discuss in details what are the Obata connection and Obata curvature later.}  In this case, there exist coordinates where three quaternionic  complex structures $(I^p)^{\ N}_M$ are  {\it constant} matrices
\p{I-HK},
in accordance with the general statement of Ref. \cite{HoPa}.

To describe a generic HKT metric, one should generalize the constraints \p{D++q=0}
and write
\be
\lb{D++q=F}
D^{++} q^{+a} \ =\ {\cal L}^{+3a} \, ,
\ee
where ${\cal L}^{+3a}$ is an arbitrary analytic superfield of harmonic charge +3, generically depending on all $q^{+b}$ and on
the harmonics.
The superfields $q^{+a}$ subject to the constraints \p{D++q=F}, \p{cons-tilde-Omega} describe {\it nonlinear}
({\bf 4}, {\bf 4}, {\bf 0}) multiplets.

The plan of the paper is the following. In the next section,
we recall the  basic definitions and properties of hyper-K\"{a}hler and HKT geometries.
In particular, we give a simple ``physical'' proof of the well-known mathematical fact that
the holonomy group of hyper-K\"{a}hler manifolds is $Sp(n)\equiv USp(2n)$
and derive the explicit expressions for  the  Obata connections of the HKT manifolds via  their complex structures.

In Sect. 3, we explain, following \cite{harm}, how  hyper-K\"{a}hler  geometries are described in the language  of  ${\cal N}{=}\,8$
supersymmetric mechanics in  one-dimensional ${\cal N}{=}\,8$ harmonic superspace obtained as $d=4 \rightarrow d=1$ dimensional
reduction of the ${\cal N}=2, d=4$ harmonic superspace formulation of the most general ${\cal N}=2, d=4$ supersymmetric hyper-K\"ahler
sigma model of Ref. \cite{harm}. Any hyper-K\"{a}hler metric
can be derived from the action \p{Sq+N8} involving the harmonic prepotential ${\cal L}^{+4}$.

In Sect. 4, we discuss the properties of generic  ${\cal N}{=}\,4$ supersymmetric $\sigma$ models with equal number
of real bosonic and fermionic dynamic variables. These models are best described in the language
of  ${\cal N}{=}\,1$  superfields ${\cal X}^M = x^M + i\theta \psi^M$. Generically, one obtains
either HKT or bi-HKT geometry \cite{GPS,Hull,biHKT}. We prove, however, an important theorem
that any model of this class with quaternionic complex structures is HKT.

In Sect. 5,  we show how to obtain generic HKT geometries in the language
 of supersymmetric  mechanical system  with the action \p{S-lin} now involving nonlinear
 supermultiplets $q^{+a}$ satisfying the constraints \p{D++q=F}.
 We show how these complicated constraints can be resolved and the metric and other geometric characteristics can be found.
 We also give a simple proof of the fact that the complex structures thus obtained are quaternionic and hence
 the geometry is indeed HKT, confirming the results of explicit calculations in Ref. \cite{DI}.

In Sect. 6, we show how these general methods work by presenting detailed consideration for a particular
example of the  hyper-K\"{a}hler Taub-NUT manifold \cite{CMP,harm}.

In the last section we discuss a general classification of HKT geometries following
from their supersymmetric description. These geometries are grouped in the families characterized
by a given constrained potential ${\cal L}^{+3a}$, but by different potentials ${\cal L}$. The geometries belonging
to a family that involves as a member  also a hyper-K\"{a}hler geometry can be called {\it reducible}
and all other HKT geometries {\it irreducible}. We note that irreducible geometries exist. In particular,
the HKT geometry derived in Ref. \cite{DV} is irreducible \cite{Valent,Papa}.

In Appendix\,A, we present the basics of ${\cal N}{=}\,4,$ $d{=}\,1$ harmonic superspace.

In Appendix\,B, we give a simple proof of the known mathematical fact that a triple
of quaternionic complex structures in tangent space, $(I^p)_{AB} = e_A^M e_B^N \, (I^p)_{MN}$,
can be brought by the tangent space rotations to the simple canonical form \p{I-HK}, \p{IJKcanon}.

In Appendices\,C and D we derive the Obata connection and the nonlinear transformation laws of the central-basis harmonic-independent
fermionic fields of the multiplet $({\bf 4, 4, 0})$  under ${\cal N}=4$ supersymmetry within the setting of Ref. \cite{DI}.

\section{HK and HKT geometries}
\setcounter{equation}0

\begin{defi}
A complex manifold is a manifold of even  real dimension endowed with the complex structure tensor $I^{\ M}_{N}$ satisfying the
conditions
\be
\lb{II=-1}
I_{M}{}^{N} I_{N}{}^{K}\  = \ - \delta_M^K \, ,
\ee
\be
\lb{Nijen}
\partial_{[M} I_{N]}{}^P \ =\ I_M{}^Q I_N{}^S \partial_{[Q} I_{S]}{}^P \, .
\ee
\end{defi}
We also assume that the manifold possesses a metric which is {\it Hermitian} with respect to the complex structure $I_M{}^{N}$
\be
I_{P}{}^{M} g_{MQ}I_N{}^{Q} = g_{PN}\,,\lb{Hermit}
\ee
which evidently amounts to the antisymmetry property
\be
I_{MN} = -I_{NM}\,. \lb{Antysymm}
\ee

The condition \p{Nijen} (it is equivalent to the requirement that the so-called Nijenhuis tensor vanishes)  provides for
{\it integrability} of the complex structure.
For an integrable complex structure,  one can
 introduce holomorphic coordinates, $x^M = \{z^m, \bar z^{\bar m} \}$, such that
the metric is manifestly Hermitian (one can always do it locally, but a nontrivial  property following from \p{Nijen}
is that the manifold can be divided
into a set of overlapping holomorphic charts with holomorphic glue functions),
\be
\lb{def-h}
ds^2 = \ 2h_{m \bar n} dz^m  d{\bar z}^{\bar n}\, .
\ee
In these coordinates,
the tensor $I_M{}^N$ has the following
nonzero components,
\be
\lb{Icompl}
I_m{}^{ n} = -I^n_{\ m} = -i \delta_m^n, \qquad I_{\bar m}{}^{\bar n} = -I^{\bar n}_{\ \bar m} = i
\delta_{\bar m}^{\bar n} \, .
\ee
It follows that $I_{m {\bar n}} = - I_{{\bar n} m} = -ih_{m{\bar n}}$.

\vspace{.2cm}

\begin{defi}
A K\"ahler manifold is a complex manifold  for which $I_{M}^{\ N}$ is covariantly constant,

\be
\label{Lev-Civ}
\nabla_P I_{M}^{\ N} = \ \partial_P I_{M}^{\ N} - \Gamma^Q_{PM} I_{Q}^{\ N} +
\Gamma^N_{PQ} I_{M}^{\ Q} \ =\  0 \, ,
\ee
where $\Gamma^Q_{PM}$ are the standard symmetric Christoffel symbols for the metric $g_{MN}$. \end{defi}

It then
follows that the K\"ahler form $K = I_{MN} \, dx^M \wedge dx^N$ is closed,
$dK = 0$. The existence of such a closed 2-form can be chosen as an alternative definition of K\"ahler manifolds.
Note that the integrability condition \p{Nijen} does not change its form if one replaces the partial derivatives
$\partial_M$ by the covariant ones with an arbitrary symmetric connection. In particular, one can replace $\partial_M$ in \p{Nijen}
by the Levi--Civita covariant derivative $\nabla_M$.  Due to \p{Lev-Civ}, Eq. \p{Nijen} is identically fulfilled.
Thus, \p{Nijen} is automatically satisfied for the covariantly constant complex structures obeying \p{Lev-Civ}.

\vspace{2mm}

{\bf Remark 1.}

For a generic complex manifold, the complex structure is not covariantly constant with respect to the usual Levi-Civita
connection appearing in \p{Lev-Civ}, but one can still define  connections with modified Christoffel symbols, such
that the modified covariant derivative of the complex structure tensor vanishes.
There are infinitely many such connections, but  a special role is played by the
{\it Bismut connection} \cite{Bismut}. This connection  involves nontrivial torsion,
\be
\lb{GamBis}
\hat \Gamma^Q_{PM} \ =\  \Gamma^Q_{PM} + \frac 12\, g^{QS} C_{SPM}
\ee
with completely antisymmetric $C_{SPM}$. Then the condition
\be
\lb{delhatI=0}
\hat {\nabla}_P I_{MN} \ = \ 0
\ee
defines $C_{SPM}$ uniquely. Its expression in holomorphic terms reads \cite{Dirac}
\be
\lb{C-Bis-hol}
C_{sp\bar m} \ =\ \partial_p h_{s \bar m} - \partial_s h_{p \bar m} , \ \ \ \ \ \ \ \
C_{\bar s \bar p  m} \ =\ \partial_{\bar p} h_{m \bar s} - \partial_{\bar s} h_{m \bar p} \, .
\ee
In other words,
\be
\lb{C-form-Bis}
C_{SPM} \, dx^S \wedge dx^P \wedge dx^M \ =\ 6(\bar \partial -  \partial) \omega \, ,
\ee
where $\omega = h_{s \bar m} dz^s \wedge d\bar z^{\bar m}$ and $\partial, \bar\partial$ are holomorphic and antiholomorphic
exterior derivatives.

For K\"ahler manifolds, the Bismut connection
coincides with the Levy-Civita connection.

\vspace{.2cm}

\begin{defi}
A hyper-K\"ahler manifold is a manifold with three different integrable
complex structures $I^p$ that satisfy the quaternion algebra
\be
\label{quatern}
I^p I^q = -\delta^{pq} +  \varepsilon^{pqr} I^r
\ee
and are subject to the covariant constancy  condition \p{Lev-Civ}.
\end{defi}

\vspace{.2cm}

Consider the  complex structures in tangent space $(I^p)_{AB} = e^M_A e^N_B  (I^p)_{MN}$.
In Appendix B we will prove that by appropriate rotations they can be reduced to  the following canonical form:

\be
\label{I-HK}
I^1 \ =\  {\rm diag} (\EuScript{I},\ldots, \EuScript{I}), \ \ \ \
I^2 \ = \ {\rm diag} (\EuScript{J},\ldots, \EuScript{J}), \ \ \ \
I^3 \ =\  {\rm diag} (\EuScript{K},\ldots, \EuScript{K}) \, ,
\ee
where $\EuScript{I}, \EuScript{J}$ and $\EuScript{K}$ are the 4-dimensional matrices related to the `t Hooft symbols $\eta^p_{AB}$
and forming quaternionic algebra:
\be
\lb{IJKcanon}
\!\!
\EuScript{I} \  =  \left( \begin{array}{cccc} 0&0&0&1 \\ 0&0&1&0 \\ 0&-1&0&0 \\ -1&0&0&0 \end{array} \right),\quad
\EuScript{J} \  =  \left( \begin{array}{cccc} 0&0&1&0 \\ 0&0&0&-1 \\ -1&0&0&0 \\ 0&1&0&0 \end{array} \right)
,\quad
\EuScript{K} \  = \left( \begin{array}{cccc} 0&1&0&0 \\ -1&0&0&0 \\ 0&0&0&1 \\ 0&0&-1&0 \end{array} \right).
\ee

Sometimes it is more convenient to represent the $4n$ vector tangent space indices $A, B$ on which the $O(4n)$ tangent space group is realized,
by the pair of indices $A, B \rightarrow (ia), (jb)$, $i, j = 1,2$, $a,b = 1, \ldots 2n$.
The indices are raised and lowered according to $X^i = \varepsilon^{ij}X_j, Y^a = \Omega^{ab} Y_b$ with $\varepsilon_{jk} = -\varepsilon^{jk}$; $\varepsilon_{12} = 1$ and
$\Omega_{ab} = -\Omega^{ab}$ being defined in \p{Omega}. In this notation,  only the subgroup $SU(2) \times Sp(n) \subset O(4
n)$
is manifest, $a, i$ being the indices of the corresponding spinorial representations. To establish the precise relation between the vector and spinor notations,
we introduce $4n$ rectangular matrices $\Sigma_A$:
\be
\label{Sigma}
(\Sigma_{1,2,3,4})^{ja} \ =\ \left( \sigma^\dagger_\mu\,, \, 0\,, \, ... \,,\, 0  \right)^{ja}, \ \ \ \
(\Sigma_{5,6,7,8})^{ja} \ =\ \left( 0\,,\, \sigma^\dagger_\mu\,, \, 0\,, \, ... \,,\, 0  \right)^{ja}, \ \ldots
\ee
with
\be
\lb{sigma-mu}
(\sigma^\dagger_\mu)^{ja} \  = \ \{(\vec{\sigma})^{ja}, -i\delta^{ja} \} \, .
\ee
In \p{sigma-mu} $a = 1,2$ and  $\vec{\sigma}$ are the standard Pauli matrices.

Then for any tensor we have the correspondence
$$
T^{\ldots {ja} \ldots} \ =\  \frac i{\sqrt{2}} \, (\Sigma_A)^{ja}\, T^{\ldots\; A \ldots}
$$
(the dots stand for all other indices).
In these terms, the flat tangent space metric is expressed as
\be
\lb{metr-simpl}
g^{ja,\, kb} \ =\  -\frac 12\, (\Sigma_A)^{ja}  (\Sigma_A)^{kb}  =
\varepsilon^{jk} \Omega^{ab}  \, ,\qquad
g_{ja,\, kb} \ =\  \varepsilon_{jk} \Omega_{ab} \, .
\ee
The symplectic constant matrix  $\Omega^{ab}$ is invariant under  $Sp(n) \subset SO(4n)$.

Note that for a real vector $V^A$, the components $V^{ja}$ obey the pseudoreality condition

\be
\lb{pseudoreal-V}
(V^{ja})^* \ =\ \varepsilon_{jk} \Omega_{ab}  V^{kb} \equiv V_{ja} \, .
\ee

\begin{thm}
The holonomy group of a hyper-K\"ahler manifold of dimension $4n$ is $Sp(n)$.
\end{thm}

This statement is very well known to mathematicians. We give  here its detailed proof in explicit ``physical'' terms.

\begin{proof}

In the spinor notation, the canonical flat complex structures \p{I-HK}, \p{IJKcanon} can be  expressed as
\be
\lb{Ip}
(I^p)^{ja, kb} \ =\ -\frac 12 (\Sigma_A)^{ja}  (\Sigma_B)^{kb} (I^p)_{AB} =
-i(\sigma^p )^{jk} \, \Omega^{ab} \, , \qquad
(I^p)_{ja, kb} \ =\ i(\sigma^p )_{jk} \,  \Omega_{ab}\,,
\ee
where
\be
\lb{symm-jk}
(\sigma^p)_{jk} \ =\ (\sigma^p)_{kj} \ =\ \varepsilon_{kl}
(\sigma^p)_{j}{}^{l}\,,\qquad (\sigma^p)^{jk} \ =\ \varepsilon^{jl}
(\sigma^p)_{l}{}^{k}
\ee
and $(\sigma^p)_l{}^{k}$ are the  standard Pauli matrices.
\footnote{Note that the matrices $(\sigma^p)^{ja}$ in \p{sigma-mu} (with both upper indices) and
$(\sigma^p)_l{}^k$ in \p{symm-jk} (with the indices placed at the different levels) coincide.  The difference
in conventions is justified by the fact that the indices in \p{symm-jk} refer to one and
the same subgroup $SU(2) \subset SO(4n)$ whereas the indices $j,a$ in \p{sigma-mu} have different nature.}

The covariant constancy condition \p{Lev-Civ} for the triplet of the complex
structures $I^p_{MN}$, after passing to the tangent space representation, takes the form
\be
\partial_P I^p_{AB}+ \left(\omega_{P, AC} I^p_{CB} - I^p_{AC} \omega_{P,  CB} \right) = 0\,, \lb{vyvod}
\ee
where
\be
\lb{omega}
\omega_{P, AB} \ =\ e_{MA} \left( \partial_P e^M_B + \Gamma^M_{PK} e^K_B \right)
\ee
is the spin connection.
Substituting the constant expression \p{Ip} for $I^p_{AB}$ in \p{vyvod}, we observe that this condition is reduced to
\be
\omega_{P, AC} I^p_{CB} - I^p_{AC} \omega_{P,  CB} =0\,,
\ee
which tells us that the spin connection understood as a matrix in tangent space, whose entries are 1-forms,  commutes with all complex structures.
The same condition in the spinor notation takes the form
\be
\lb{CCC}
(\omega_M)_{ia \, jb}(\sigma^p)^{j}_{\ k} - (\sigma^p)^{\ j}_i  (\omega_M)_{ja \,  kb}= 0\,.
\ee

A generic antisymmetric connection $(\omega_M)_{ia \, jb} = -(\omega_M)_{jb \, ia}$ can be parametrized as
\be
\lb{2terms-om}
(\omega_M)_{ia \, jb} \ =\ \varepsilon_{ij} T_{M\,(ab)} \ + \ B_{M\,[ab]\, (ij)} \ =\
\varepsilon_{ij} T_{M\,(ab)} \ + \ B_{M\,[ab]}^q (\sigma^q)_{ij}
\ee
with arbitrary $ T_{M\,(ab)}$  and $B_{M \,[ab]}^q$. When one substitutes this into \p{CCC}, the first term $\propto \varepsilon_{ij}$ does not contribute and we are led to

\be
\lb{komu-tator}
B_{M \,[ab]}^q \, [\sigma^q, \sigma^p]_{ij} \ =\ 0 \, .
\ee
This holds for any $p$, which implies
\be
B_{M\,[ab]}^q\  = 0\,.
\ee

We thus derived,
\be
\lb{om-sympl}
\omega_{AB} \ =\ \omega_{ia \, jb} \ =\ \varepsilon_{ij} (T)_{(ab)} \, .
\ee
But any symmetric matrix of dimension $2n$ can be presented as
\be
\lb{sympl-alg}
T_{(ab)} \ = \ T_a^{\ c} \, \Omega_{bc} \, ,
\ee
where $T_a^{\ c} \in sp(n)$.~\footnote{Indeed, an element $h$ of $sp(n)$ is a Hermitian $2n$-dimensional
matrix satisfying $h^T \Omega + \Omega h = 0$.} Thus, $\omega_A^{\ B}$ belongs to $sp(n)$. But then
$R_A^{\ B} = d\omega_A^{\ B} + \omega_A^{\ C} \wedge \omega_C^{\ B}$ also belongs to $sp(n)$, and the theorem is proven.
\end{proof}

\vspace{.2cm}

{\bf Remark 2.}

One can also easily prove the inverse theorem: {\it If the holonomy group is $Sp(n)$, i.e. $R_A^{\ B} \in
sp(n)$, then one can choose three quaternionic covariantly constant complex structures and the manifold is hyper-K\"ahler.}

\begin{proof}
Basically, it follows from the following lemma:
{\it Let $\mathfrak{g}$ be a Lie algebra and $\mathfrak{h}$ be its subalgebra. Let $\hat A_M $ and $\hat F_{MN}$ be the gauge potential and the field density for
the algebra $\mathfrak{g}$.  Let $ \hat F \in \mathfrak{h}$. Then one can always choose the gauge where also
$\hat A \in \mathfrak{h}$. }

Such a gauge is well-known, it is the Fock-Schwinger gauge $x^M A_M = 0$ \cite{Fock}. In this gauge the
potential is expressed via the field density,
\be
\lb{A-via-F}
\hat A_M = - \int_0^1 d\alpha \, \alpha x^N \hat F_{MN}(\alpha x) \, .
\ee
In the  case of interest, $\mathfrak{g} = so(4n), \ \mathfrak{h} = sp(n), \omega \equiv A$ and $R \equiv F$. If $R \in sp(n)$ one can choose the coordinates
and vielbeins with $\omega \in sp(n)$.
And once $\omega \in sp(n)$, it commutes with the quaternionic complex structures \p{Ip}. Bearing in mind \p{vyvod}, it follows that
the convolutions of these flat structures with the vielbeins are covariantly constant.
\end{proof}

This means that one can {\it define} a hyper--K\"ahler manifold as a
manifold where the Riemann curvature form $R_A^{\ B}$ lies in the $sp(n)$
algebra. This definition and  Definition 3 are equivalent.

\vspace{.1cm}

\begin{defi}
An HKT manifold is a manifold endowed with three integrable quaternionic
complex structures that are covariantly constant with respect to one and the
same Bismut connection.

\end{defi}

The curvature form $\hat R_A^{\ B}$ of this Bismut connection belongs to
$sp(n)$ --- it is proven in exactly same way as for the
Riemann curvature form for hyper-K\"ahler manifolds. Alternatively,
if there exists a torsionful metric-preserving connection whose
curvature form lies in $sp(n)$, one can find three quaternionic complex
structures that are covariantly constant with respect to this connection, and we
are dealing with an HKT manifold.

Besides the Bismut connection, a distinguished role for HKT manifolds is played  by the Obata connection.

\begin{defi}
The Obata connection is a torsionless connection with respect to which all three quaternionic
complex structures of an HKT manifold are covariantly constant.
\end{defi}

For a hyper-K\"ahler manifold, the Obata connection coincides with the Levy-Civita connection, but it is not
so in a generic HKT case. The essential difference is that the covariant Obata derivative of the metric
tensor does not vanish! This means in particular
that vectors do not only rotate under parallel
transports, but also change their length; the holonomy group is not compact and complicated.

\begin{thm}
Let $I,J,K$ be three integrable quaternionic complex structures,
\be
\lb{quat}
IJ = -JI = K, \ \ \ \ \ JK = -KJ = I, \ \ \ \ \ KI= -IK = J \, .
\ee
Choose the complex coordinates associated with $I\,,$ i.e. assume that $I\,,$ is constant and diagonal as is given in Eq. \p{Icompl}. Then the Obata connection is given by the formula \cite{Michel}
\be
\lb{Obata}
(\Gamma^O)^k_{mn} \ =\  J_n^{\ \bar l} \partial_m J_{\bar l}^{\ k}   \ =\   K_n^{\ \bar l} \partial_m K_{\bar l}^{\ k}  , \ \ \ \
(\Gamma^O)^{\bar k}_{\bar m \bar n} \ =\ [(\Gamma^O)^k_{mn}]^* \, ,
\ee
and all other components vanish.
\end{thm}

\begin{proof} It consists of four steps

\begin{lem} In the chosen coordinates, the only non-vanishing components of the structure $J$ are
$J_{\bar l}^{\ k}$ and $J_l^{\ \bar k}$. The same is true for $K$.
\end{lem}

\begin{proof}
Introduce the operator
$\iota$ which acts on a generic $n$-form $\omega$ according to the rule:
 \be
\lb{iota}
&&{\rm if} \ \ \ \ \ \ \ \ \ \ \omega \ =\ \frac 1{n!}\, \omega_{M_1
\ldots M_n} dx^{M_1} \wedge \cdots \wedge dx^{M_n}\,, \nn
&&{\rm then} \ \ \ \ \iota \omega \ =\ \frac 1{(n-1)!}\,
\omega_{N [M_2 \ldots M_{n-1}} (I)^{N}_{\ M_1]}\,.
dx^{M_1} \wedge \cdots \wedge dx^{M_n}\,.
\ee
For a form $\omega_{p,q}$ with $p$ holomorphic and $q$
antiholomorphic indices, the action
of $\iota$ is reduced to the multiplication by $i(p-q)$.

There is a 2-form associated with each complex structure. Define
$$ {\cal J} \ =\ J_{MN} dx^M \wedge dx^N, \ \ \ \ \ \ {\cal K} \ =\ K_{MN} dx^M \wedge dx^N $$
and consider the form ${\cal J}+i{\cal K}$. Using the definition \p{iota} and the quaternion algebra \p{quat}, it is straightforward to verify
that $\iota ({\cal J}+i {\cal K}) = 2i({\cal J}+i{\cal K})$. That means that the form ${\cal J} + i {\cal K}$ has type $(2,0)$ with respect to $I$.
Analogously, $\iota ({\cal J}-i{\cal K}) = -2i({\cal J}-i{\cal K})$, so that ${\cal J}-i{\cal K}$ is of
type $(0,2)$. It follows that the only non-vanishing components of $J,K$ have either
both holomorphic or both anti-holomorphic lower indices, and {\it Lemma 1} is proven. \end{proof}

We have introduced the operator  \p{iota} to make contact with \cite{GPS,Grant} and to use it later  in \p{iota-C}. But {\it Lemma 1} can actually be proven without
resorting to the language of forms.  It directly follows from the quaternion algebra \p{quat} with the special choice \p{Icompl} for $I_M^{\ N}$.
Indeed, the relations \p{quat} imply
\be
\lb{quat+-}
I(J+iK) = -i(J+iK), \ \ \ \ \ \ \ (J+iK)I = i(J+iK), \nn
I(J-iK) = i(J-iK), \ \ \ \ \ \ \ (J-iK)I = -i(J+iK) \, .
\ee
Bearing in mind  \p{Icompl}, we derive that the only non-vanishing
components of the
tensor $J+iK$ are $(J + iK)_m^{\ \bar n}$ and the only non-vanishing
components of the  tensor $J-iK$ are $(J - iK)_{\bar m}^{\  n}\,$. From this {\it Lemma 1} immediately follows. It follows in addition that
\be \lb{JKcomp}
K_m^{\ \bar n} = -i J_m^{\ \bar n}\,, \quad
K_{\bar m}^{\ n} = i J_{\bar m}^{\ n}\,.
\ee
Then the second relation in \p{quat} amounts to the basic property \p{II=-1} of the complex structures $J,K$. 

\begin{lem}
The expression \p{Obata} is symmetric under permutation $m \leftrightarrow n$.
\end{lem}

\begin{proof}
This follows from integrability. Indeed, the condition \p{Nijen} for the structure $J$ implies
$$
J_S{}^{M} (\partial_M J_N{}^{K} - \partial_N J_M{}^{K} ) \ =\
J_N^{\ Q} (\partial_Q J_S{}^{K} - \partial_S J_Q{}^{K} ) \, .
$$
Choose $S=s, N=n, K = k$. Then, bearing in mind {\it Lemma 1},
$$
J_s{}^{\bar m} \partial_n J_{\bar m}{}^{k} \ =\ J_n{}^{\bar m} \partial_s J_{\bar m}{}^k \, .
$$
\end{proof}

Using $J^2 = -1$ to flip the derivatives, we may derive
\be
\lb{sym-Obata}
J_{\bar m}{}^k \partial_n J_s{}^{\bar m} \ =\ J_{\bar m}{}^k \partial_s J_n{}^{\bar m} \, .
\ee
It amounts to the simple relation
\be
\lb{SimpleRel}
\partial_n J_s{}^{\bar m} - \partial_s J_n{}^{\bar m} = 0\, ,
\ee
and the same holds for $K$.
In the language of forms, this means that the exterior holomorphic derivative of the
$(2,0)$-form ${\cal J} + i {\cal K}$ vanishes.

\vspace{2mm}

{\bf  Remark 3.}

In fact, the existence of such a closed
holomorphic $(2,0)$-form may be taken as an {\it alternative definition} of an HKT-manifold \cite{Grant,Verb}.
The existence of a universal Bismut covariant derivative, as is spelled out in  Definition 4, can be derived from that.

\begin{lem}
The Obata covariant derivatives of all complex structures vanish.
\end{lem}
\begin{proof} It can be checked rather directly using Eqs.~\p{Obata}, \p{Icompl},  and {\it Lemma 1}. We leave it to the reader.
\end{proof}

\begin{lem}
The Obata connection is unique.
\end{lem}

\begin{proof}
Suppose there are two different Obata connections. Let their difference be
$\Delta^N_{PM}$. The identities
\be
\Delta^Q_{PM} I_Q^{\ N} - \Delta^N_{PQ} I_M^{\ Q}  \ =\ 0 \, , \lb{Del-I} \\
\Delta^Q_{PM} (J\pm iK)_Q^{\ N} - \Delta^N_{PQ} (J \pm iK)_M^{\ Q}  \ =\ 0 \lb{Del-JK}
\ee
should hold.

It follows from \p{Del-I}, \p{Icompl} and from the symmetry $\Delta^N_{PM} = \Delta^N_{MP}$
that all the components except $\Delta^n_{pm}$ and $\Delta^{\bar n}_{\bar p \bar m}$ vanish.
And the latter vanish due to \p{Del-JK} and {\it Lemma 1}.

\end{proof}

\end{proof}

\section{Classification of hyper-K\"ahler manifolds with harmonic tools}
\setcounter{equation}0

The standard de Rham complex is characterized by a nilpotent exterior derivative operator $d$ and its Hermitian
conjugate $d^\dagger$. In supersymmetric approach, $d$ and $d^\dagger$ are mapped into a complex supercharge $Q$ and its conjugate.
The ${\cal N}=2,\, d=1$ superspace description involves several real $({\bf 1}, {\bf 2}, {\bf 1})$ superfields,
\be
\lb{121superfield}
X^M = x^M + \theta \psi^M + \bar \psi^M \bar \theta + F^M \theta \bar \theta \, ,
\ee
where $x^M$ are the bosonic dynamical variables (the coordinates on the manifold), $\psi^M$ are complex dynamical fermionic
variables and $F^M$ are the bosonic auxiliary fields (in the component expansion of the action
\p{act-Rham} they enter without time derivatives and can be integrated over). The action has the form
\be
\lb{act-Rham}
S \ =\ \frac 12 \int dt d\theta d\bar\theta \, g_{MN}(X) D X^M \bar D X^N \, ,
\ee
where $D$ and $\bar D$ are the covariant supersymmetric derivatives. Going down into components, one obtains the standard bosonic
kinetic part
of the Lagrangian,
\be
\lb{Lkin-bos}
L_{\rm bos} \ =\ \frac 12 \, g_{MN}(x) {\dot x}^M {\dot x}^N
\ee
describing the motion of a particle along a curved manifold. The action \p{act-Rham} can be written for any manifold.

In special cases, in addition to the manifest ${\cal N} = 2$ supersymmetry that the action \p{act-Rham} exhibits,
one can observe the presence of extra ``hidden'' supersymmetries. Thus, if the manifold is K\"ahler (and, hence, even-dimensional), the action \p{act-Rham}
is invariant under the extra supersymmetry,
\be
\lb{delX-K}
\delta X^M \ =\ I_{N}{}^M(X) \,  \left(\epsilon D X^N - \bar \epsilon \bar D X^N \right) \, ,
\ee
where $I_{N}{}^M$ is the complex structure and $\epsilon$ is a complex Grassmann transformation parameter. The total
supersymmetry is then ${\cal N} = 4$. The relevant superalgebra  closes off shell, which can be directly checked by evaluating Lie brackets of the superfield
transformations \p{delX-K}.

Besides the formulation in terms of $2n$ $({\bf 1, 2, 1})$ multiplets, two other off-shell formulations of this model are possible. It can
be formulated in terms of $n$ pairs of complex chiral ${\cal N}=2$ superfields $({\bf 2, 2, 0})$ and  $({\bf 0, 2, 2})$ \cite{Dirac}.
The same model can also be formulated in extended ${\cal N}=4, d=1$ superspace in terms of  $n$ chiral
$({\bf 2, 4, 2})$ superfields \cite{Zumino}.

For hyper-K\"ahler manifolds, the action is invariant, besides the manifest ${\cal N} = 2$ supersymmetry, with respect
to three different extra supersymmetries \p{delX-K} involving three
quaternionic complex structures $I^p$ and three different complex Grassmann parameters $\epsilon^p$. One can show that the generators
of different supersymmetries anticommute. The total supersymmetry of the model is thus ${\cal N} = 8$. One can further
prove that a hyper-K\"ahler metric is not only a sufficient, but also necessary condition to have ${\cal N} = 8$ supersymmetry
in the action \p{act-Rham} \cite{AG-F}.

The observation that the action \p{act-Rham} with three extra ${\cal N}=2$ supersymmetries picks up the hyper-K\"ahler
manifolds as the bosonic targets, does not yet give any tool of how to {\it explicitly construct} hyper-K\"ahler metrics. The latter can only
be achieved in the harmonic superspace approach, where {\it all} eight one-dimensional supersymmetries are manifest
and off-shell.
\footnote{Note that the harmonic superspace description implies the presence of an infinite number of auxiliary fields.
This is a crucial distinction of hyper-K\"ahler models from the  simple K\"ahler models discussed above:  the latter
can be described by  ${\cal N}=2$  or ${\cal N}=4$ superfields that live in ordinary superspace and
involve a finite number of auxiliary fields.}

Consider an extended  ${\cal N} {=}\, 8,$ $d{=}\,1$ superspace $(t, \theta_{i\alpha}, \bar \theta^{i\alpha})$,
$\bar \theta^{i\alpha}=(\theta_{i\alpha})^\ast$ with $\alpha = 1,2$
and the G-analytic superspace $(\zeta, u) \equiv (t_A, \theta^{+}_{\alpha}, \bar \theta^{+\alpha} , u)$.
The latter superspace is obtained from  the  ${\cal N}{=}\,2$, $d{=}\,4$ harmonic analytic superspace \cite{harm} by dimensional
reduction. It represents a
direct generalization of the ${\cal N} {=}\, 4,$ $d{=}\,1$ harmonic analytic superspace \cite{IvLe} briefly
described  in Appendix A (we only endow the odd coordinates with the extra
index $\alpha$). Consider a G-analytic superfield
\footnote{We have not displayed here the index $A$ for $t$, as we did not do so  in Eq.\p{q+lin}.}
 \be
\lb{q+N8}
Q^+(\zeta, u) = F^+(t, u)  +  \theta^{+}_{\alpha} \chi^\alpha(t,u) + \bar\theta^{+\alpha}  \kappa_\alpha(t,u)
+ \theta^{+}_{\alpha} \bar \theta^{+\alpha} A^-(t, u) + (\theta^{+}_{\alpha} \bar \theta^{+\alpha})^2D^{-3}(t,u) +\ldots
\ee
(only {\it some}  terms relevant for us in what follows  are displayed in the expansion).
Take $2n$ such superfields $Q^{+a}$ subject to the constraint
\p{cons-tilde-Omega}.
Consider the action (see \cite{CMP} and Chapter 5 of Ref.~\cite{harm}\footnote{The coefficients in \p{q+N8} and \p{Sq+N8} display
some deviations
from those in \cite{harm}. The measure of Grassmann integration over the analytic superspace is defined in Appendix A.} )
\be
\lb{Sq+N8}
S \ =\ \int dt \, du \, d^2 \theta^+ d^2 \bar \theta^+ \left[
\,\frac 12 \
Q^+_a D^{++} Q^{+a}  +  {\cal L}^{+4}(Q^+, u) \right] \, ,
\ee
with
\be
\lb{D++N8}
D^{++} = \ \partial^{++} + 2i\theta^{+}_{\alpha} \bar \theta^{+\alpha} \frac {\partial}{\partial t}\, ,
\ee
$\partial^{++}$ being defined in Eq. \p{part-pmpm}.
Here  ${\cal L}^{+4}$ is an arbitrary function of $Q^{+a}$ and $u^\pm$, such that it carries the harmonic charge +4. By construction,
this action has a manifest ${\cal N} = 8$ supersymmetry (there are four complex transformation
parameters associated with the shifts of $\theta^{+}_{\alpha}$ and   $\bar \theta^{+\alpha}$:
$\delta \theta^{+}_{\alpha} = \epsilon_{\alpha}^{i}u^+_i$,
$\delta \bar \theta^{+\alpha} = \bar{\epsilon}^{\alpha}_{i}u^{+ i}$).

The superfield equation of motion following from \p{Sq+N8} reads
\be
\lb{eqmotL+4}
D^{++} Q^{+a} \ =\  \Omega^{ab} \frac {\partial {\cal L}^{+4}}{\partial Q^{+b}} \, .
\ee

Substituting there the expansion \p{q+N8}, we obtain a set of the equations for the components. In particular, we derive
\be
\lb{eqmot-F}
\partial^{++} F^{+a} \ =\ \Omega^{ab} \frac {\partial {\cal L}^{+4}}{\partial F^{+b}}
\ee
and
\be
\lb{eqmot-A}
\partial^{++} A^{-a}  + 2i\dot{F}^{+ a} -   \Omega^{ab} \frac {\partial^2 {\cal L}^{+4}}{\partial F^{+b} \partial F^{+c}} A^{-c}
\ \ + \ \ \ \ \mbox{(terms  with  fermionic  fields)} =0\,.
\ee

Consider equation \p{eqmot-F}. If ${\cal L}^{+4}$ were absent, the linear equation $\partial^{++} F^{+a} = 0$ would have a simple solution $F^{+a} \equiv F^{+a}_0 = x^{ja} u^+_j $, with $x^{ja}$ obeying the pseudoreality condition
\p{pseudoreal-V}  following from the constraint \p{cons-tilde-Omega}.
When ${\cal L}^{+4} \neq 0$, it is rather difficult task to find the solution to  \p{eqmot-F}. To date, it was found in a closed form
only for a few particular choices of ${\cal L}^{+4} \neq 0$ \cite{CMP,GIOT,harm} including the choice corresponding
to the Taub-NUT manifold discussed below.  In the general case, the solution to  \p{eqmot-F} can be found by iterations:
attribute a factor $\lambda$ to ${\cal L}^{+4}$ and represent  the solution as a formal series
\be
\lb{series-lambda}
F^{+a} \ =\ F^{+a}_0 + \lambda F^{+a}_1 +  \lambda^2 F^{+a}_2 + \cdots\,.
\ee
We obtain the chain of equations
\be
\lb{chain}
\partial^{++}  F^{+a}_1  &=& \Omega^{ab} \, \frac {\partial {\cal L}^{+4}} {\partial F_0^{+b}}\, , \nn
\partial^{++}  F^{+a}_2  &=& \Omega^{ab} \, \frac {\partial^2 {\cal L}^{+4}} {\partial F_0^{+b}  \partial F_0^{+c}} F^{+c}_1 \, , \nn
\ldots &=& \ldots\,.
\ee
These are in fact algebraic equations, as becomes clear if one expands their left-hand and right-hand sides in a proper harmonic basis.
For example, we represent
\be
F_1^{+a} \ =\ A^{(ijk)a} u^+_i u^+_j u^-_k +  B^{(ijklp)a} u^+_i u^+_j u^+_k  u_l^- u_p^- + \cdots
\ee
(the linear term $\propto u^+_j$ does not contribute in the left-hand sides of \p{chain}; it is attributed to
$F_0^+$) and
\be
\Omega^{ab}\, \frac {\partial {\cal L}^{+4}}{\partial F^{+b}}  \ =\  C^{(ijk)a} u^+_i u^+_j u^+_k
+ D^{(ijklp)a} u^+_i u^+_j u^+_k  u_l^+ u_p^-   +  \cdots\,.
\ee
Then the first equation implies
$$
A \ =\ C, \qquad B \ =\ \frac 12\, D \, , \qquad {\rm etc.}
$$
We first solve the equation for $F_1^{+a}$, then we substitute its solution to the equation for $F_2^{+a}$, solve it, substitute into the equation for $F_3^{+a}$, etc. As a result, $F^{+a}(t, u)$ is expressed via the harmonic-independent coefficients $x^{ja}(t)$, which have the meaning
of the coordinates on the hyper-K\"ahler  manifolds that we are set to describe. Note that the pseudoreality
conditions \p{pseudoreal-V} for $x^{ja}$ imply that the vectors $x^M \ =\  i(\Sigma^M)^{ja} x_{ja}/\sqrt{2}$
[with constant matrices $\Sigma^M \equiv \Sigma_A$
defined in Eq.\p{Sigma}] are real. But any other choice of $4n$ real coordinates $x^M$ is possible.

If we suppress the fermion dependence, the
equation \p{eqmot-A} can be solved in a similar way. The solution
${\tilde A}^{-a}(t, u)$ of this truncated equation is also expressed via $x^{ja}(t)$. It follows from \p{eqmot-A}  that
\be
\tilde{A}^{-a}\ =\ -2i \dot{x}^{ia} u^-_i  + {\rm nonlinear \ terms}\,. \lb{A-1}
\ee

To find the metric of the manifolds of interest, we substitute the solutions thus obtained for $F^{+a}(t, u)$ and
${\tilde A}^{-a}(t, u)$ into the bosonic part of the action \p{Sq+N8}. Indeed, if one expressed the latter via the components, one obtains a
very simple expression
\be
\lb{S-comp-Q+}
S \ =\ \frac{i}{2}\,\int \,dt \,du \,  {\tilde A}^-_a \dot{F}^{+a} \,,
\ee
where the other bosonic components of the superfield \p{q+N8} do not contribute! After expressing ${\tilde A}^-_a$ and ${F}^{+a}$
through $x^{ia}(t)$, the action \p{S-comp-Q+} takes the generic form
\be
S \ =\ \int \, dt\, \frac12\, g_{ia ,\,kb}\, \dot{x}^{ia}\dot{x}^{kb}\,, \quad  g_{ia ,\,kb} = \varepsilon_{ik}\Omega_{ab} + O(\lambda)\,. \lb{bosN8}
\ee

A similar program can be carried out for the fermionic components. Everything can be expressed through the lowest component
$\psi^{\alpha a}(t)$ in the harmonic expansion of  $\chi^{\alpha a}(t,u)$.
[The variables $\kappa^{a}_{\alpha}$ are not independent,
but can be expressed via $\chi^{\alpha a}$ in virtue of \p{cons-tilde}].
We have altogether $4n$ complex dynamic fermionic variables ---
one complex fermionic variable for each real bosonic coordinate.
Their bilinear contribution to the Lagrangian
has the structure $\propto \bar \psi \dot{\psi}$.
The variables can be chosen such that the coefficients in the fermion kinetic term
and in the bosonic kinetic term are given by the same metric tensor. There is also a four-fermionic term, with
the coefficient proportional to the Riemann tensor.

One thus obtains the action of a supersymmetric $\sigma$ model, a particular case of the generic action \p{act-Rham} (with the auxiliary fields being eliminated).
In view of the theorem proven in \cite{AG-F}, the presence of ${\cal N} = 8$ supersymmetry dictates
the metric to be hyper-K\"ahler. In particular, if ${\cal L}^{+4} = 0$, the metric is flat.

However, the road from the superspace action \p{Sq+N8} to the metric is long and stony. We have outlined above the exact
regular procedure to derive the metric and express it as an infinite series over the formal expansion parameter $\lambda$.
The existence of such a procedure implies that a unique solution exists. As was already mentioned, a closed analytic solution to the
equations \p{eqmotL+4} was  obtained so far only in a few particular cases.
In Sect. 6, we will show how the explicit solution can be found for the Taub-NUT manifold.

We explained how to construct a hyper-K\"ahler metric, based on an arbitrary function
${\cal L}^{+4}(Q^+, u)$
 of harmonic charge +4. A legitimate question is whether {\it any} hyper-K\"ahler metric can be derived in this way?

The answer to this question is positive. In the paper \cite{OgIv}
(see also Chapter 11 of the book \cite{harm}), the problem
was solved in a different way --- not invoking supersymmetry, but solving instead the constraint $R_A^{\ B} \in sp(n)$ (see Theorem 1 and the remark after it)
It was shown that a {\it general} solution to this constraint depends on an arbitrary harmonic function
${\cal L}^{+4}(Q^+, u)$
and that this solution coincides with the solution following from \p{Sq+N8}.

The last remark of this Section yet concerns the superfield equation \p{eqmotL+4}.  Besides the kinematical
equations \p{eqmot-F}, \p{eqmot-A}  and similar equations for fermionic fields, it encompasses as well the dynamical equations
for fields $F^+, \xi^\alpha$ and $\kappa_\alpha$.
In particular, it contains the equation (with fermionic fields suppressed)
\be
\partial^{++}D^{-3} + 2i \dot{\tilde{A}}^{-a} + {\rm nonlinear \ terms} = 0\,,
\ee
which, in virtue of  the expressions \p{A-} and $F^{+ a} = x^{i a} u^+_i +\ldots$, implies
\be
\ddot{x}^{ja} + {\rm nonlinear \ terms} = 0\,.
\ee
In the ${\cal N}=4\,, d=1$ supersymmetric description of hyper-K\"ahler sigma models which we will discuss below, the equation
like \p{eqmotL+4} becomes a harmonic constraint which does not impose any dynamical restrictions
on the involved fields.

\section{HKT and bi-HKT supersymmetric $\sigma$ models}
\setcounter{equation}0

Supersymmetric $\sigma$ models considered in the previous section involved a complex fermionic field for
each real bosonic coordinate. There is another class of models with half as much fermionic degrees of
freedom; they contain a real fermion for each real bosonic coordinate. These models can be described in terms of
$({\bf 1}, {\bf 1}, {\bf 0})$ superfields living in ${\cal N} = 1$ superspace with only one real $\theta$ coordinate,
\be
\lb{N1field}
{\cal X}^M = x^M + i\theta \psi^M \, .
\ee
A generic action bringing about the structure $\sim g {\dot x}^2$ in the bosonic sector reads
\be
\lb{SgenN1}
S \ =\ \frac i2 \int dt d\theta \ g_{MN}({\cal X}) {\dot {\cal X}}^M D {\cal X}^N - \ \frac 1{12}
\int dt d\theta \ C_{SPM} D {\cal X}^S D {\cal X}^P D {\cal X}^M \, ,
\ee
where
\be
\lb{DN1}
D \ =\ \frac {\partial}{\partial \theta} - i\theta \frac {\partial}{\partial t}
\ee
is the ${\cal N} = 1$ supersymmetric covariant derivative; $D^2 = -i\partial_t$.
 The symmetric tensor $g_{MN}$ gives the metric and the
antisymmetric $C_{SPM}$ gives the torsion.

The corresponding component Lagrangian is
\be
\lb{L(N=1)comp}
L \ =\ \frac 12\, g_{MN} {\dot x}^M {\dot x}^N + \frac i2\, g_{MN}
\psi^M \nabla\psi^N - \frac{1}{12}\, \partial_K C_{SPM}
\psi^K \psi^S \psi^P \psi^M \, ,
\ee
where
\be
\lb{nabla(C)}
\nabla\psi^M={\dot \psi}^M + \hat \Gamma^M_{PS} {\dot x}^P \psi^S \, ,
\ee
\be
\lb{Gamma(C)}
\hat \Gamma_{N, PS} = g_{MN}\hat \Gamma^M_{PS} = \Gamma_{N, PS} + \frac 12 \, C_{NPS} \, .
\ee

The  ${\cal N} = 1$ supersymmetry of the action \p{SgenN1} is manifest, the components of ${\cal X}^M$ transform
as
\be
\delta x^M = i\epsilon_0 \psi^M\,, \quad  \delta \psi^M = -\epsilon_0 \dot{x}^M\,.\lb{epsilon0comp}
\ee
However, we are interested in the models including at least
two real supercharges --- their presence is necessary for a model to
enjoy nontrivial dynamical constraints including double degeneracy in the spectrum of the
Hamiltonian.

Thus, we require the action to be invariant under the following extra supersymmetry transformations:
\be
\lb{delXN2}
\delta {\cal X}^M \ =\ \epsilon I_{N}^{\ M} D{\cal X}^N \,,
\ee
where $I_N^{\ M}({\cal X}^M)$ is some tensor to be specified below. The components in \p{N1field} are transformed as
\be
\lb{delXN2-comp}
\delta {x}^M &=& i\epsilon I_{N}^{\ M} \psi^N \, , \nn
\delta \psi^M & =& \epsilon \left( I_{N}^{\ M} \dot{x}^N  - i \partial_S  I_{N}^{\ M} \psi^S \psi^N \right)\, .
\ee
We also require that the commutator of two such transformations boils down to the time translation. Then the square of the generator of the transformations \p{delXN2-comp}
coincides with the Hamiltonian, and we obtain the minimal ${\cal N} = 2$ supersymmetry algebra:
\be
\label{SUSY-alg}
(Q_1)^2 = (Q_2)^2 = H, \ \ \ \ \ \{Q_1, Q_2\} = 0 \, .
\ee

It remains to figure out the restrictions on $I_N^{\ M}$ which ensure the fulfillment of these requirements.
One can make the following statement \cite{GPS}:

\begin{thm}
The action \p{SgenN1} is invariant under \p{delXN2} and the algebra \p{SUSY-alg} holds
if the following set of conditions is satisfied:

\begin{enumerate}

\item  $I^2 = - \mathbb{1}$ as in \p{II=-1}.

\item
 $I_N{}^{M}$ is integrable and satisfies \p{Nijen}.

\item The matrix $I_{MN}=g_{NK}I_M{}^{K}$ is skew-symmetric:
\be
\lb{asym-I}
I_{MN}=-I_{NM}= I_{[MN]}\, .
\ee

\item $I_N^{\ M}$  satisfies the condition
\be
\lb{sym-nabla-I}
\hat \nabla_L I_N^{\ M} +  \hat \nabla_N I_L^{\ M} \ =\ 0 \, ,
\ee
where $\hat \nabla_L$ is the covariant derivative with the torsionful affine connection \p{Gamma(C)}.

\item There is an extra condition on the torsion tensor $C$ that can be represented in the  language of forms as
\be
\lb{iota-C}
\iota d C= \frac{2}{3}\,d (\iota C)\, ,
\ee
where the operator $\iota$ was defined in \p{iota}.

\end{enumerate}
\end{thm}

In the original paper \cite{GPS} this theorem was proved by explicit component calculations.
We give here a somewhat simpler proof based on the language of ${\cal N} = 1$ superfields.

\begin{proof}

The conditions 1 and 2 follow from the algebra \p{SUSY-alg}. Note first that
\be \delta (D {\cal X}^N) \ =\ -\epsilon D(I_L^{\ N} D {\cal X}^L ) \ =\
-\epsilon (\partial_K I_L{}^{N} ) D {\cal X}^K D {\cal X}^L +
i \epsilon I_L{}^{N} {\dot {\cal X}}^L  \, .
 \ee
The commutator of two supersymmetry transformations \p{delXN2} is then derived to be
\begin{eqnarray}
\lb{com-s-X}
\left(\delta_2\delta_1-\delta_1\delta_2 \right)\delta {\cal X}^M &=&
2i\epsilon_1\epsilon_2 (I^2)_K{}^{M} \dot{\cal X}^K \\
&&
+2\epsilon_1\epsilon_2\Big[I_K{}^{L}\left(\partial_L I_N{}^{M}\right) +\left(\partial_N I_K{}^{L}\right)I_L{}^{M}\Big]\,D{\cal X}^K D{\cal X}^N \, . \nonumber
\end{eqnarray}
If we want it to coincide with $-2i\epsilon_1\epsilon_2\, \partial_t {\cal X}^M$  [as is dictated by Eq.\p{SUSY-alg} )] the conditions \p{II=-1}, as well as
\begin{equation}
\lb{eq-intr}
\left(\partial_L I_{[N}{}^{M}\right)I_{K]}{}^{L} +\left(\partial_{[N} I_{K]}{}^{L}\right)I_L{}^{M}=0
\end{equation}
follow. Using \p{II=-1}, the condition \p{eq-intr} can be brought into the form
\p{Nijen}.

The conditions {\it 3-5} follow from the vanishing of the variation of the action under \p{delXN2}. The calculation gives
\begin{eqnarray}
\lb{var-sterms}
\delta S & =& \epsilon \int dt d\theta \, I_{(MN)} \, {\dot {\cal X}}^M  {\dot {\cal X}}^N
-\frac{i\epsilon}{4} \int dt d\theta \, P_{M, SN} \, {\dot {\cal X}}^M  D {\cal X}^S  D {\cal X}^N  \nn
&&
+  \frac {\varepsilon}{12} \int dt d\theta \, T_{RSNM}  D {\cal X}^R  D {\cal X}^S  D {\cal X}^N  D {\cal X}^M\, ,
\end{eqnarray}
where
\be
\lb{PandT}
&& P_{M, SN} = P_{M, [SN]} = 2\nabla_M I_{[SN]} - C^P{}_{SN} I_{ML}  - 2\nabla_S I_{(MN)} + 2\nabla_N I_{(MS)}\,, \nonumber \\
&& T_{RSNM} = T_{[RSNM]} = \left(\partial_L C_{[SNM}\right) I_{R]}{}^{L} -3 C_{L[SN} \left(\partial_R I_{M]}{}^{L}\right).
\ee
Note that $\nabla_M$ entering \p{PandT} are the ordinary Levy-Civita covariant derivatives.

Let us concentrate on the second term in \p{var-sterms}. We represent
\be
P_{M,[SN]} \ =\ P_{[MSN]} + \frac 13 \left(2P_{(M,[SN]}  + P_{S,[MN]} -   P_{N,[MS]}\right)
\ee
and note  the identity
\be
\lb{po-chastjam}
\int dt d\theta\, P_{[MSN]} {\dot {\cal X}}^M  D {\cal X}^S  D {\cal X}^N \ =\ - \frac i3
\int dt d\theta\,  \partial_R P_{[MSN]}  D {\cal X}^R D {\cal X}^M D {\cal X}^S  D {\cal X}^N \, .
\ee
To  derive \p{po-chastjam} , one has to trade
$\dot{{\cal X}}^M$ for $i D^2{{\cal X}}^M$ and integrate by parts.).

We thus present the variation as a sum of three linearly independent structures
\be
\lb{var-sterms-1}
\delta S & =& \epsilon \int dt d\theta \, I_{(MN)} \, {\dot {\cal X}}^M  {\dot {\cal X}}^N
-\frac{i\varepsilon}{12} \int dt d\theta \, (2 P_{M, [SN]}+ P_{S, [MN]} - P_{N, [MS]} ) \, {\dot {\cal X}}^M  D {\cal X}^S  D {\cal X}^N  \nn
&&
+  \frac {\varepsilon}{12} \int dt d\theta \, \left(T_{RMSN}  - \partial_{[R} P_{MNS]}\right) D {\cal X}^R  D {\cal X}^M  D {\cal X}^S  D {\cal X}^N\, .
\ee

It vanishes provided
\begin{equation}
\lb{I-sym}
I_{(MN)} = 0\,,
\end{equation}
\begin{equation}
\lb{P-sym}
2 P_{M, [SN]}+ P_{S, [MN]} - P_{N, [MS]} = 0
\end{equation}
and
\begin{equation}
\lb{TP-cond}
T_{RMSN}= \partial_{[R} P_{MSN]}\,.
\end{equation}

Eq. \p{I-sym} gives the requirement 3 in the list above. The condition \p{P-sym}, after taking account of \p{I-sym} and after substituting  $P_{M, [SN]}$ from Eq. \p{PandT}, yields
\be
2 \hat{\nabla}_{M} I_{[SN]}+ \hat{\nabla}_{S} I_{[MN]} - \hat{\nabla}_{N} I_{[MS]} = 0\,. \lb{Equivii}
\ee
By symmetrizing over $M \leftrightarrow S$ and raising the index $N$ (we are allowed to do so, bearing in mind that $\hat{\nabla}_{S}\, g^{MN} = 0$), we arrive at \p{sym-nabla-I}.

Finally,
substituting into \p{TP-cond} the expression for $T_{RMSN}$  from  \p{PandT} and
\be
P_{[MSN]} &=& \frac23\left(\partial_M I_{[SN]} - \frac12 C^L_{\; MS}I_{[NL]} + {\rm cycle}\,(M,S,N)\right) \, , \lb{SymmP}
\ee
we derive
\be
\big( \partial_L C_{[MNS} \big) I^{\;L}_{R]} + I^{\;L}_{[S} \big( \partial_{R} C_{MN]L} \big)  - 2 C_{L[MN}\partial_R I^{\;L}_{S]} = 0\, , \lb{iota-C2}
\ee
which coincides with \p{iota-C}.

\end{proof}

The conditions \p{II=-1}, \p{Nijen} and \p{I-sym} imply that the tensor $I_{MN}$ has all the properties of the
complex structure and can be interpreted as such. It is natural then to expect that the geometry thus obtained is
a complex geometry and the algebra \p{SUSY-alg} maps into the classical Dolbeault complex, with the complex supercharges $Q_1 \pm i Q_2$
being mapped into the holomorphic exterior derivative operator $\partial$ and its
Hermitian conjugate $\partial^\dagger$.

This guess is {\it almost} correct. In fact, a generic action \p{SgenN1} defines a twisted Dolbeault complex involving extra
{\it holomorphic} torsions \cite{twistDolb}. When such torsions are present, the Hamiltonian does not commute anymore with the
fermion charge operator. Such systems have been studied in \cite{Hull,real} in the language of ${\cal N} = 2$ superfields.

To understand how the holomorphic torsions appear in the ${\cal N} = 1$ language used in this paper, we prove the following
theorem\footnote{The {\it inverse} statement was proven in Sect. 5.1 of Ref. \cite{CKT-OKT}.}

\begin{thm}
Let $(g,I,C)$ satisfy  conditions 1--5 above. Consider the holomorphic decomposition of the torsion form with respect to the complex structure $I$,
\be
\lb{hol-decomp}
C \ =\ C_{3,0} + C_{2,1} + C_{1,2} + C_{0,3} \, .
\ee
Then the mixed part $C_{2,1} + C_{1,2}$ is the Bismut torsion \p{C-form-Bis} for the complex structure $I$.
The holomorphic and antiholomorphic parts are closed,
\be
\lb{zamkn}
\partial C_{3,0} = \bar \partial C_{(0,3)} \ =\ 0 \, .
\ee
\end{thm}

\begin{proof}
Let us express \p{sym-nabla-I} in complex coordinates. The complex structure acquires then a simple form \p{Icompl}. Choose
$L=l, N = \bar n, M = \bar m$. Only the second term in \p{sym-nabla-I} is left, and we obtain
$$ \hat \Gamma^{\bar m}_{\bar n l}  = 0 \ \Longrightarrow \ \hat \Gamma_{s, \bar n l} =  0 \ \Longrightarrow
\ C_{sl\bar n} = \partial_l h_{s \bar n} - \partial_s h_{l \bar n} \, , $$
which gives $C_{2,1} = -6\partial \omega$ as in \p{C-form-Bis}, indeed. To derive $C_{1,2} = 6\bar \partial \omega$, one should
choose $L=l, N= \bar n, M=m$, which leads to the condition $\hat \Gamma^m_{n \bar l} = \hat \Gamma_{\bar m, n\bar l} = 0$.
\footnote{When all three indices in \p{sym-nabla-I} are holomorphic or all of them are antiholomorphic, the equality
is fulfilled identically. The choice $L=p, N=n, M = \bar m$ does not bring about new information.}

Consider now the condition  \p{iota-C2} and choose all the free indices holomorphic: $L = l, M = m, N =n, S=s$. Bearing in mind \p{Icompl}, we immediately derive $\partial_{[r} C_{mns]} = 0$, which means
that $C_{(3,0)}$ is closed. Choosing all the indices antiholomorphic, we derive the closedness of
$C_{(0,3)}$.
\end{proof}

It is worth noting that the condition \p{iota-C2} involving both holomorphic and antiholomorphic indices does not yield any new information.
 For the mixed torsion components  \p{C-Bis-hol} [and such form follows, as we have just seen, from  \p{sym-nabla-I}],
 it is identically satisfied.

Note also that in the ${\cal N} = 2$ language used in \cite{Hull,real},
the (anti)holomorphic
components $C_{3,0}$ and $C_{0,3}$ in the torsion are associated with the presence of extra (anti)holomorphic terms in the action:
\be
\lb{B-hol}
\Delta S =  \int dt d\theta d\bar\theta \, {\cal B}_{nm} DZ^n DZ^m \ + \ {\rm c.c.}
\ee
with antisymmetric ${\cal B}_{mn}$. Then $C_{3,0} \propto  \partial {\cal B}$, $C_{0,3} \propto \bar \partial \bar {\cal B}$,
from which \p{zamkn} follows.

\vspace{.2cm}

We go over now to ${\cal N} = 4$ models. These models should possess {\it three} extra supersymmetries of the type
\p{delXN2}. Each of the complex structures $I^{p=1,2,3}$ should be integrable, and the constraints \p{sym-nabla-I},
\p{iota-C} should be satisfied.
There are two extra constraints following from the requirement that three new
supercharges together with the supercharge $Q$ associated with the
explicit ${\cal N} = 1$ supersymmetry of the action \p{SgenN1} satisfy the standard  ${\cal N} = 4$ superalgebra,
\be
\lb{SUSY-algN4}
Q^2 = H, \qquad \{Q^p, Q^q \} =  2 \delta^{pq} H, \qquad \{Q, Q^p \} = 0 \, .
\ee
The new constraints are
\be
\lb{Clifford}
I^p I^q + I^q I^p \ =\ -2 \delta^{pq}\, ,
\ee
\be
\lb{concom}
(I^p)_{[M}^{\ \ S} \partial_S (I^q)_{N]}^{\ \ L} -
\partial_{[M} (I^q)_{N]}^{\ \ S}  (I^p)_S^{\ \ L}   \ + (p \leftrightarrow q) \ =\ 0 \, .
\ee
The first condition says that the complex structures $I^p$ satisfy the Clifford algebra.
The second condition is the vanishing
of the so-called {\it Nijenhuis concomitant}.

Clifford complex structures in \p{Clifford} are not necessarily quaternionic,
$I^1 I^2 \neq I^3$ etc. As was noticed in \cite{biHKT} (see {\bf Proposition 6} there), the closure
of the multiplication algebra represents in this case a direct sum ${\mathcal H}_+ +
{\mathcal H}_-$  of two quaternion algebras. Indeed, define
\be
J^p = \frac 12\, \varepsilon^{pqr} I^q I^r, \qquad \Delta = - I^1 J^1 = -I^2 J^2 = -I^3 J^3 \, .
\ee
One can then observe that the both algebras ${\mathcal H}_\pm$ involving the generators
\be
I^p_\pm \ =\ \frac 12 (I^p \pm J^p), \qquad \Delta_\pm \ =\ \frac 12 (1 \pm \Delta)
\ee
are closed and quaternionic. The operators $\Delta_{\pm}$ play the role of the corresponding quaternion unities.

Bearing that in mind, one can show that a generic triple of Clifford complex structures projected in tangent space can be  chosen in the form
\be
\label{I-biHKT}
I^1 &=& {\rm diag} (\underbrace{\EuScript{I},\ldots, \EuScript{I}}_{n^*},\ \underbrace{-\EuScript{I}, \ldots,
-\EuScript{I}}_{m^*}), \qquad
I^2 \ = \ {\rm diag} (\underbrace{\EuScript{J},\ldots, \EuScript{J}}_{n^*}, \
\underbrace{-\EuScript{J}, \ldots, -\EuScript{J}}_{m^*}), \nn
I^3 &=& {\rm diag} (\underbrace{\EuScript{K},\ldots, \EuScript{K}}_{n^*},
\underbrace{ -\EuScript{K} , \ldots, -\EuScript{K}}_{m^*}) \, ,
\ee
with $\EuScript{I}, \EuScript{J}$ and $\EuScript{K}$ written in \p{IJKcanon}.

We see that a generic ${\cal N} = 4$ model involves two sectors associated with the subspaces of dimension $4n^*$ and $4m^*$.
Suppose that either $n^*$ or $m^*$ vanish and there is only one such sector with quaternionic complex structures.
We can prove an important theorem:
\begin{thm}
If the complex structures entering the laws of transformation \p{delXN2} in a generic
${\cal N} =4$ model with the action \p{SgenN1} satisfying the conditions of Theorem 3 are quaternionic,
these structures are covariantly constant with respect to the universal
Bismut connection and the manifold is HKT.
\end{thm}

\begin{proof}

The integrability allows us to choose complex coordinates associated with any
of the complex structures. Let us do so for $I$. Then the only
nonzero components of the tensor $I_M^{\ N}$ are displayed in Eq.\p{Icompl}.
Consider the tensors $J \pm iK$. As follows from Lemma 1 for
Theorem 2 [specifically, from the relations \p{JKcomp}], their only nonzero components are $(J+iK)_m^{\ \bar n} = 2J_m^{\ \bar n}$
and $(J-iK)_{\bar m}^{\ n} = 2J_{\bar m}^{\ n}$.

Note that the quaternionic algebra \p{quat} and the properties $J^2 = K^2 = - \mathbb{1}$ imply
\be
\lb{J+J-}
(J - iK)_{\bar m}^{\ n} (J+iK)_n^{\ \bar l} = -4 \delta_{\bar m}^{\bar l} \, ,
\nn [5pt]
(J + iK)_{n}^{\ \bar l} (J-iK)_{\bar l}^{\ m} = -4 \delta_n^m \, .
\ee
Considering the constraint \p{sym-nabla-I} for the complex structure $I$, we derived earlier
(see the proof of Theorem
4) the properties
\be
\lb{Gam-p-qbar}
\hat \Gamma^n_{l \bar s} = \hat \Gamma^{\bar n}_{\bar l s} \ =\ 0 \, .
\ee

Consider now the constraint
\be
\lb{sym-nabla-JK}
\hat \nabla_L (J+iK)_N^{\ M} +  \hat \nabla_N (J+iK)_L^{\ M} \ =\ 0 \, .
\ee

We choose $L=l$ and $M=m$, but do not specify the holomorphicity of $N$. Using the fact
that the only nonzero components of $J+iK$ have an antiholomorphic upper index, we see that in this case
most
of the terms in \p{sym-nabla-JK} vanish
and we derive
\be
\hat \Gamma^m_{N \bar s} (J+ iK)_l^{\ \bar s} + \hat \Gamma^m_{l \bar s} (J+ iK)_N^{\ \bar s} \ =\ 0 \, .
\ee
However, the second term in the left-hand-side vanishes due to \p{Gam-p-qbar}, and the constraint
boils down to $ \hat \Gamma^m_{N \bar s}
(J+ iK)_l^{\ \bar s}  = 0 $. Multiplying this by $(J- iK)_{\bar r}^{\ l}$ and using
the first relation in \p{J+J-}, we arrive at the constraint
\be
\lb{Gam-M-qbar}
\hat \Gamma^m_{N \bar r} \ =\ 0
\ee
for all $N$, holomorphic and antiholomorphic. The constraint
$\hat \Gamma^{\bar m}_{s N}  = 0$ can be derived in a similar way.

Consider now the covariant derivative without symmetrization,
\be
\lb{nabla-Gamma}
\hat\nabla_L I_N^{\ M} \ =\
\hat \Gamma_{LN}^Q I_Q^{\ M} - \hat \Gamma_{LQ}^M I_N^{\ Q}
\ee
and choose again $M=m$ (the case $M = \bar m$ is treated analogously).
We note that, if $N=n$, the right-hand-side of
\p{nabla-Gamma} vanishes identically due to the cancelation of the two terms and if $N= \bar n$, it boils down to
$-2i \hat \Gamma^m_{L \bar n}$, which vanishes due to \p{Gam-M-qbar}.

We have proven that $\hat\nabla_L I_M^{\ N} = 0$. Choosing the complex coordinates
associated with $J$ or with $K$, we can repeat all the arguments and prove that
$\hat\nabla_L J_M^{\ N} = \hat\nabla_L K_M^{\ N} = 0\,$.

\end{proof}

{\bf Remark}. We have proved this theorem for a generic action \p{SgenN1} including the second torsion term.
If we started from the  action that satisfied the conditions of Theorem 3, but
 included only the first term in \p{SgenN1}, the complex structures would be covariant with respect to the
ordinary Levi-Civita connection and we would arrive to hyper-K\"ahler geometry.

\vspace{1mm}

Going back to generic models with non-vanishing $n^*$ and $m^*$, it is natural to call them
{\it bi-HKT models}. Another nomenclature for the same class of models is ``Clifford-K\"ahler-with-Torsion'' (CKT) sigma
models \cite{Hull}. In Ref.\cite{biHKT}, a restricted class of such models involving ordinary and ``mirror''
linear ${\cal N}=4$  multiplets $({\bf 4, 4, 0})$
was studied in detail. The explicit expressions
for the Lagrangian, Hamiltonian and supercharges were derived.
If one suppresses there the dynamic variables in one of the sectors, e.g. the coordinates in the
subspace of dimension $4m^*$ and their fermionic superpartners, the reduced
 Lagrangian describes an HKT model associated
with linear $({\bf 4, 4, 0})$ multiplets. But there is also a nontrivial interaction between the sectors, and the
Lagrangian is not just the sum of two independent HKT Lagrangians. The simplest such model has a 8-dimensional target space.
It was first discussed in \cite{IvLeSu} and studied in detail in \cite{CKT-OKT}.
 The models of this type can also be described in terms of
${\cal N}=2, d=1$ superfields, as it was done in Sect. 7 of Ref.
\cite{DI}. In this case, the action is expressed in terms of ordinary and mirror (or twisted) chiral superfields.

\section{Harmonic description of HKT models}
\setcounter{equation}0

\subsection{Generalities}

In Sect. 3, we outlined how hyper-K\"ahler geometries can be described
in terms of ${\cal N} = 8$, $d=1$ supersymmetric $\sigma$ models.
The  superfields naturally
realizing  ${\cal N} = 8\,, \,d=1$ supersymmetry off shell are harmonic
superfields.

On the other hand,  ${\cal N} = 4$ off-shell superfields with an equal number
of real dynamical bosonic and fermionic degrees of freedom can be defined both
in the conventional superspace and harmonic superspace. For {\it linear}
$({\bf 4}, {\bf 4}, {\bf 0})$ multiplets, the conventional description works quite
well, it allows one to derive the metric and all other geometric
characteristics of interest \cite{IvLeSu,CKT-OKT,FS-HKT,biHKT}. But for {\it nonlinear} multiplets,
the harmonic description is much more convenient.

As was already mentioned in the Introduction, we start from the  ${\cal N} = 4$ superspace $(t, \theta_i , \bar \theta^i)$
and harmonize it as outlined in Appendix A. We consider then a set of $2n$ G-analytic superfields
$q^{+a}(t, \theta^+, \bar \theta^+, u) \equiv q^{+a}(\zeta, u)$. Their component expansion reads
\be
\lb{Expq+}
q^{+ a}(\zeta, u) = f^{+ a}(t,u) + \theta^+ \chi^a(t, u) + \bar\theta^+ \kappa^a (t, u)
+ \theta^+\bar\theta^+ A^{- a}(t,u)\,.
\ee
It is much shorter than the component expansion of a  ${\cal N} = 8$ superfield $Q^+$ in Eq.~\p{q+N8}.

We impose then the pseudoreality constraint \p{cons-tilde}, which implies in particular
$$
\kappa^a(t, u^\pm_j) = \bar \chi^a (t, \widetilde{u^\pm_j}) = \bar \chi^a (t, u^{\pm j}) \, ,
$$
and the harmonic constraint \p{D++q=F}. Note the essential difference between  \p{D++q=F} and \p{eqmotL+4}.
The latter is a superfield equation of motion following from the action \p{Sq+N8}. But in the  ${\cal N} = 4$ case, the
analytic superspace includes only two odd variables, $\theta^+$ and $\bar \theta^+$, and the structure
${\displaystyle \int} d^4\theta^+ {\cal L}^{+4} $ does not exist.~\footnote{One can add to the action the term
$\sim {\displaystyle \int} d^2\theta^+ {\cal L}^{++}$,
but this amounts to the inclusion of the gauge fields living on the manifold \cite{IvLe,FIL,Maxim,Maxim1}.
In this paper, such an option will not be considered. }
Thus, Eq.\p{D++q=F} is an {\it external} nonlinear constraint.

The superfield constraint \p{D++q=F} amounts to the following constraints on
the components:
\be
\lb{Bos}
\partial^{++} f^{+ a} &=& {\cal L}^{+3 a}(f^+, u^\pm) \,,  \\
{\cal D}^{++} \chi^a &=& {\cal D}^{++} \bar\chi^a = 0\,, \lb{ferm} \\
{\cal D}^{++} A^{- a} &=&  -2i \dot{f}^{+ a} + \ \frac
{\partial^2 {\cal L}^{+3 a}}{\partial f^{+b} \partial f^{+c} }\,\bar\chi^b\chi^c\,,
\lb{A-}
\ee
where the action of the covariant harmonic
derivative ${\cal D}^{++}$ on any contravariant symplectic vector $G^a$  is defined as
\be
\lb{D++contra}
{\cal D}^{++} G^a = \partial^{++}G^a -   E^{+ 2 a}_{b} G^b \, , \quad E^{+ 2 a}_{b} := \frac{\partial{\cal L}^{+3 a}}{\partial f^{+ b}}.
\ee

We also need for further uses to define the action of  ${\cal D}^{++}$ on  covariant symplectic vectors $H_a$.
\footnote{The use of the terms
 ``covariant derivative'' for the operator ${\cal D}^{++}$,
``covariant and contravariant symplectic vectors'' is justified bearing in mind the covariance of the
constraints \p{Bos}--\p{A-} under certain {\it analytic diffeomorphisms}
  \be
\lb{anal-diff}
\delta f^{+ a} = \lambda^{+a} (f^+,u)\,, \ \ \delta \chi^a =
(\partial_{+b} \lambda^{+a}) \chi^b\, , \ \  \delta A^{-a} = \cdots\,, \ \
\delta {\cal L}^{+3a} = \cdots \, , \ \
 \ee
 (see Ref. \cite{DI}
for  details).}  We define
\be
\lb{D++co}
{\cal D}^{++} H_a = \partial^{++} H_a + E^{+ 2 b}_{a} H_b\,.
\ee
Then
\be
\lb{GaHa}
{\cal D}^{++} (G^a H_a) = \ \partial^{++} (G^a H_a) \, .
\ee

Note that
generically the contravariant and covariant vectors are not obtained from one another
by  multiplying by $\Omega^{ab}$ or $\Omega_{ab}$. The relation
\be
\lb{G=Om-G}
G^a \ =\ \Omega^{ab} G_b
\ee
is not always compatible with \p{D++contra} and \p{D++co}.
However, for hyper-K\"ahler manifolds, where ${\cal L}^{+3a}$ satisfies
\p{L3=dL4} below,  it {\it is} compatible and we can assume it to
hold. We will do so in the next section.
We also note the useful identity
\be
\lb{D++dot}
{\cal D}^{++} {\dot f}^{+ a} = 0 \, .
\ee
It is simply the time derivative of the constraint \p{Bos}.

The constraints \p{Bos}--\p{A-} are rather complicated, but in some simple cases they can be
resolved analytically. We will do so in Sect.~6 for the Taub-NUT metric, but in this section
we concentrate on the general structure of the model. We note that, after the
constraints are resolved, the harmonic dependence of all component fields is
fixed and everything is expressed via the pseudoreal dynamical bosonic fields $x^{ia}(t)$ [or real $x^M(t)$]
and their fermionic superpartners $\psi^{ia}(t)$ [or $\psi^M(t)$]. For a linear multiplet, this relationship is very simple [see
\p{q+lin}], but in the nonlinear case the expressions are more complicated.

We write the action in the form \p{S-lin} with an arbitrary ${\cal L}$.
Substituting there the expansion \p{Expq+}, imposing the constraints
\p{Bos}-\p{A-} and solving them, we finally obtain a ${\cal N} = 4$
supersymmetric model expressed via the dynamical variables
$x^{ia}(t) \equiv x^M(t)$ and $\psi^{ia}
(t)  \equiv \psi^M(t)$.

The models of this kind were discussed in the previous section.
We saw that they admit either an HKT or a bi-HKT geometry. Our case is
more restrictive, however. One can prove  the following important theorem:

\begin{thm}
The superfield action \p{S-lin} with the constraints \p{Bos}--\p{A-}
describes an HKT geometry, which is reduced to hyper-K\"ahler geometry in some special cases.
\end{thm}

\begin{proof}
It is sufficient to prove that the complex structures in this model are
quaternionic and then use the result of Theorem 5.
Indeed, the component action respects ${\cal N}=4$ supersymmetry by
construction because it is obtained from the superfield action.
Hence  all the conditions needed for the ${\cal N}=1$ action \p{SgenN1}
to be ${\cal N}=4$  supersymmetric and for Theorem 5
to be applicable are satisfied.

The complex structures can be obtained by deriving the law of supertransformations for the variables $x^{ja}$ and comparing it with
\be
\lb{trans-x-psi}
\delta x^{ja} \ =\ i\epsilon_p \, (I^p)_{kb}{}^{ja} \, \psi^{kb} \, ,
\ee
which is the first line in \p{delXN2-comp}
for three complex structures with the indices in spinor notation.

Supertransformations of the dynamical variables are generated by the shifts of the odd superspace
coordinates $\theta$. Consider the G-analytic superfield \p{Expq+}.
The shifts  $\delta \theta^+ = \epsilon^+, \, \delta \bar\theta^+ = \bar\epsilon^+$  [where
$\epsilon^+ = \epsilon^{{k}} u^+_{{k}}, \, \bar\epsilon^+ = \bar \epsilon^{{k}} u^+_{{k}}$;  $(\epsilon_{{k}})^* =
\bar \epsilon^{\u{k}}$]  yield \footnote{The transformation rules of the component fields
are found from the generic superfield transformation law $\Phi'(Z) \simeq \Phi( Z +\delta Z) = \Phi(Z)
+ \delta Z\,\partial_Z \,\Phi(Z) + \ldots \,$.}
\be
\lb{trans-f+}
\delta f^{+a}(t,u) \ =\ \epsilon^+ \chi^a(t,u) +\bar \epsilon\bar \chi^a (t, \tilde{u}) \, .
\ee
We need, however, to derive the transformation law for harmonic-independent fields.
The field $f^{+a}(t,u)$ can, indeed, be expressed in terms of a
harmonic-independent ``central basis'' field $x^{ja}(t)$ after
solving equation \p{Bos}, as was explained in detail in Sect. 3.
 Equations \p{ferm} can be resolved in a similar way.
We may represent their solution as
\be
\lb{psi-most}
\chi^a(t, u) \ =\ (M^{-1})^a_{\u{b}} \, \psi^{\u{b}}(t) \, , \qquad  \bar\chi^a(t, \tilde{u})
\ =\ (M^{-1})^a_{\u{b}} \, \bar\psi^{\u{b}}(t)  \, ,
\ee
where $M^{\u{b}}_a$ and its inverse $(M^{-1})^a_{\u{b}} $ are very important objects called the {\it bridges}.
They relate the fields carrying the {\it world} symplectic index $a$
to those carrying the {\it tangent space} symplectic index
$\u{a}$.
\footnote{From now on we will distinguish the ordinary
world spinorial and symplectic  indices $j, a$
from the  tangent space underlined indices $\u{j}, \u{a}$.
The operator ${\cal D}^{++}$ acts,
according to \p{D++contra} and \p{D++co}, only on the world symplectic indices.}
The bridges play the role of vielbeins for the
analytic diffeomorphisms \p{anal-diff}.
 They satisfy the equations
\be
\lb{D++most}
{\cal D}^{++}  \,    M_a^{\u{b}}  = 0   \qquad \Longleftrightarrow \qquad  {\cal D}^{++}  \,
  (M^{-1})^a_{\u{b}}  = 0 \, .
\ee

  The harmonic-independent fields $\psi^{\u{b}}$ and $\bar\psi^{\u{b}}$
can be joined into the  quartet
$\psi^{\u{kb}}$ ,
\be
\lb{combining-psi}
\psi^{\u{1b}} \equiv \psi^{\u{b}}, \ \ \ \ \ \ \  \psi^{\u{2b}}  \equiv   \bar  \psi^{\u{b}}
\, ,
\ee
which satisfies the same pseudoreality conditions \p{pseudoreal}
as the bosonic coordinates, bearing in mind the convention \p{conj-chia}.
It is natural to interpret  $\psi^{\u{kb}}$ as the fermion field
carrying the tangent space index $A$ and expressed in spinor notation.
Note, however, that there are many solutions to \p{D++most}, interrelated by the right multiplications.
\footnote{This is the so-called ``$\tau-gauge$ freedom'' discussed in \cite{DI}:
\be
\lb{M->RM}
M^{\u{b}}_a \ \rightarrow \ M^{\u{c}}_a \, R_{\u{c}}^{\ \u{b}} \, ,
 \qquad (M^{-1})^a_{\u{b}} \ \rightarrow \ (R^{-1})_{\u{b}}^{\u{c}}
(M^{-1})^a_{\u{c}}.
\ee
 For the ``analytic diffeomorphism vielbeins'' $M$ and $M^{-1}$, the transformations \p{M->RM} play
 the same role as the orthogonal tangent space rotations
for the ordinary vielbeins $e^M_A$.
Note that we changed the notation $L^{\u{b}}_{\u{a}} \rightarrow R^{\u{b}}_{\u{a}}$, compared to \cite{DI},
 not to mix up this matrix with various Lagrangians.}
and the definition of  $\psi^{\u{kb}}$ depends on the choice of the matrix $R$
We will see below that
one can always choose the matrix  $R$ in such a way that  $\psi^{\u{kb}}$ defined above exactly coincides with
$i (\Sigma_A)^{\u{kb}}\psi^A/\sqrt{2}$, where $\psi^A = e^A_M\, \psi^M$.

Bearing \p{psi-most} in mind, the transformation law \p{trans-f+} can be rewritten as
\be
\lb{Mdelx=psi}
M_a^{\u{b}}\, (\partial_{kb} f^{+a})\, \delta x^{kb}(t) \ =\ \epsilon^+ \psi^{\u{b}}(t) + \bar \epsilon^+  \bar
\psi^{\u{b}}(t) \, ,
\ee
where $\partial_{kb} = \partial/\partial x^{kb}$, $(\psi^{\u{a}})^* = \bar \psi_{\u{a}}$ and
$\epsilon^+ = \epsilon^{\u{k}} u^+_{\u{k}}, \, \bar\epsilon^+ = \bar \epsilon^{\u{k}} u^+_{\u{k}}$.

The right-hand side of \p{Mdelx=psi} is proportional to $u_{\u{j}}^+$
while more complicated terms in the harmonic expansion are absent.
The same should be true for the left-hand side. Indeed, differentiating \p{Bos}
over $x^{kb}$, we obtain
$$
{\cal D}^{++} (\partial_{kb} f^{+a}) = 0 \, .
$$
[Cf. Eq.\p{D++dot}; note that in the equation above the operator
 ${\cal D}^{++}$ does not ``feel'' the presence of the world index $kb$.]
Using \p{GaHa} and
\p{D++most}, we derive
$$
\partial^{++} \left[   M_a^{\u{b}} \, (\partial_{kb} f^{+a})  \right] \ =\ 0 \, ,
$$
and hence the object $E^{+\u{b}}_{kb} :=  M_a^{\u{b}} (\partial_{kb} f^{+a}) $
 can be represented as \footnote{For convenience, we define
the vielbeins with the opposite sign compared to those in \cite{DI}.}
\be
\lb{E+propu+}
E^{+\u{b}}_{kb} \ =\ e^{\u{kb}}_{kb} (t) u^+_{\u{k}\,.}
\ee
The transformation law \p{Mdelx=psi} acquires the form
\be
\lb{delx=epsi}
\delta x^{kb} \ =\ e^{kb}_{\u{kb}} (\epsilon^{\u{k}} \psi^{\u{b}} +
\bar \epsilon^{\u{k}} \bar \psi^{\u{b}}) \, ,
\ee
where $e^{kb}_{\u{kb}}$ is the inverse of $e^{\u{kb}}_{kb}$.
We will see by the end of this section that these matrices are
nothing but the ordinary vielbeins under a particular choice of
the matrix $R$ in \p{M->RM}.

We now define the matrix $\epsilon_{\u{l}}{}^{\u{k}}$  according to
\be
\lb{combining}
\epsilon_{\u{1}}{}^{\u{k}}  \equiv  \epsilon^{\u{k}} , \ \ \ \ \ \ \epsilon_{\u{2}}{}^{\u{k}} \equiv
\bar \epsilon^{\u{k}}
\ee
and represent it as
\be
\lb{eps-kl}
\epsilon_{\u{l}}{}^{\u{k}}  \ =\ i\epsilon_0 \delta^{\u{k}}_{\u{l}} + \epsilon_p  (\sigma^p)_{\u{l}}{}^{\u{k}}
\ee
with real $\epsilon_0, \epsilon_p$. The matrix \p{eps-kl} satisfies the identity
$( \epsilon_{\u{l}}{}^{\u{k}})^* = \varepsilon_{\u{ki}} \, \varepsilon^{\u{lj}} \,   \epsilon_{\u{j}}{}^{\u{i}}$\, .
In this notation, the transformation law \p{delx=epsi} is expressed as
\be
\lb{delx=epsi-matr}
\delta x^{kb} \ =\ e^{kb}_{\u{kb}} \, \epsilon_{\u{l}}{}^{\u{k}} \, \psi^{\u{lb}} \ =\
i\epsilon_0 \, e^{kb}_{\u{kb}} \, \psi^{\u{kb}} \, +  \, \epsilon_p  \,
 e^{kb}_{\u{kb}} \, (\sigma^p)_{\u{l}}{}^{\u{k}}  \psi^{\u{lb}}
\, .
\ee
If we define
\be
\psi^{kb} \equiv e^{kb}_{\u{kb}} \psi^{\u{kb}}\,, \lb{Defpsi}
\ee
the first term in \p{delx=epsi-matr} reads
simply as $ \delta x^{kb} =  i\epsilon_0 \psi^{kb}$. This is a ${\cal N} = 1$
 supersymmetry transformation
indicating that $x^{kb}$ and $\psi^{kb}$ are superpartners: they
 represent the components of the ${\cal N} = 1$ superfield \p{N1field} [cf.
\p{epsilon0comp}].

The second term in \p{delx=epsi-matr} can be represented as
\be
\lb{trans-tripl}
\delta x^{kb} \ =\  i\,\epsilon_p \, e^{\u{lc}}_{lc}\,(I^p)_{\u{lc}}{}^{\u{kb}} \, e^{kb}_{\u{kb}}\psi^{lc} \,
\equiv  i\,\epsilon_p \,(I^p)_{lc}{}^{kb}\psi^{lc} \, ,
\ee
where

\be
\lb{I-updown}
(I^p)_{\u{lc}}{}^{\u{kb}} \, \ =\ -i \,(\sigma^p)_{\u{l}}{}^{\u{k}} \, \delta_{\u{c}}^{\u{b}} \, .
\ee
Lowering the indices in \p{I-updown}, we reproduce the flat quaternionic complex structures \p{Ip}.
The complex structures $(I^p)^{\ M}_N = e_{NA} (I^p)_{AB} \, e^M_B$, or those in the spinor notation,
\be
(I^p)_{lc}{}^{kb} = -i e^{\u{lc}}_{lc} (\sigma^p)_{\u{l}}{}^{\u{k}} e^{kb}_{\u{kc}}\,, \lb{StrSpin}
\ee
which enter \p{trans-tripl}, are also quaternionic. Hence, according to  Theorem 5,  our manifold is HKT.

\end{proof}

This result was earlier achieved in  Ref.\cite{DI}  in another way, by the explicit calculations
of the  connections and the torsions. As we have now seen, one can arrive
at this conclusion in a simpler way,  merely by inspecting the
laws of supersymmetry transformations and
the corresponding complex structures.

\subsection{The metric}

We now describe, following \cite{DI}, in more technical detail how the metric of the HKT manifold corresponding to the action \p{S-lin}
with the constraint \p{D++q=F} is derived. To this end,
we have to express the action (it is sufficient to look at the bosonic action) in terms
of the harmonic-independent coordinates $x^{ja}(t)$. This is achieved in several steps.

At the first step, we express the bosonic part of the action \p{S-lin} in harmonic components $f^{+a}(t, u)$ and
$\tilde{A}^{-a}(t,u)$ --- the component $A^{-a}$ with suppressed fermion dependence. Using the expansion \p{Expq+} and the expansion
\be
\lb{q-a}
q^-_a &=& D^{--} q^+_a  \nn
&=&  \partial^{--} f^+_a + 2i \theta^- \bar \theta^- {\dot f}^+_a + \theta^+ \bar \theta^+ \partial^{--} A^-_a
+ (\theta^- \bar \theta^+ + \theta^+ \bar \theta^-) A^-_a   -2i \theta^+ \theta^- \bar \theta^+ \bar \theta^- \dot{A}^-_a \nn
&& +\, \theta^- \chi_a + \bar \theta^- \bar \chi_a + \theta^+ \partial^{--} \chi_a + \bar \theta^+ \partial^{--} \bar \chi_a
+ 2i \theta^- \bar \theta^- \theta^+ {\dot \chi_a} + 2i \theta^- \bar \theta^-\bar \theta^+ \dot {\bar\chi}_a\,,
\ee
we derive
\be
\lb{act-via-fA}
S^{\rm bos} = \frac{i}4 \int du dt \Big[ 2{\dot f}^{+a} \tilde{A}^{-b} \, \partial_{+[a} \partial_{-b]} {\cal L}
+  \left({\dot f}^{-a} \tilde{A}^{-b} - {\dot f}^{+a} \partial^{--} \tilde{A}^{-b} + \frac12\,\tilde{A}^{-a} \tilde{A}^{-b}
\right)  \partial_{-a} \partial_{-b} {\cal L} \Big].
\ee

At the second step, we resolve the constraint \p{Bos} [or \p{D++dot}, which is more convenient] and
the constraint
\be
\lb{tildeA-}
{\cal D}^{++} \tilde{A}^{-a} = -  2i \dot{f}^{+a}
\ee
and express ${\dot f}^{+a}(t,u)$ and $\tilde{A}^{-a}(t,u)$
via ${\dot x}^{ja}$.
Using the same arguments as that we used when deriving \p{E+propu+}, we obtain
\be
\lb{f+sol}
{\dot f}^{+a} \ =\ (M^{-1})^a_{\u{b}} \, e^{\u{kb}}_{kb} \, {\dot x}^{kb} \, u^+_{\u{k}} \, .
\ee
Here $e^{\u{kb}}_{kb}$ are the same vielbeins as in \p{E+propu+}, as is clear if we represent
$$
{\dot f}^{+a} \ =\ (\partial_{kb}  f^{+a}) {\dot x}^{kb} \, .
$$

The solution to the constraint \p{tildeA-} reads
\be
\lb{A-sol}
\tilde{A}^{-a} \ =\  -2i(M^{-1})^a_{\u{c}} \, e^{\u{lc}}_{lc} \, {\dot x}^{lc} \, u^-_{\u{l}} \, .
\ee
One can explicitly check it by acting on \p{A-sol} with the operator ${\cal D}^{++}$ and using \p{D++most} and the identity $\partial^{++} u^-_{\u{l}} =   u^+_{\u{l}}$.

Now note that $\tilde{A}^{-a}$ can also be presented as
\be
\lb{A-|and|E-2}
\tilde{A}^{-a} \ =\ -2i ({\dot f}^{-a} - E^{-2 a}_b {\dot f}^{+b} ) \, ,
\ee
where $f^{-a} = \partial^{--} f^{+a}$ and
\be
\lb{E-2}
E^{-2 a}_b \ =\ M^{\u{a}}_{b} \, \partial^{--} (M^{-1})^a_{\u{a}}\,.
\ee
This can be verified by acting with the operator ${\cal D}^{++}$ on the right-hand side of \p{A-|and|E-2}
and using the definitions \p{D++contra}, \p{D++co}, \p{E-2}, \p{D++most} together with the useful identity following from \p{E-2}
$$
{\cal D}^{++} E^{-2 a}_b = \partial^{--}E^{+ 2 a}_b \,,
$$
where $E^{+ 2 a}_b$ was defined in \p{D++contra}. Also note the relation
\be
{\cal D}^{--}\tilde{A}^{-a} = \partial^{--}\tilde{A}^{-a} - E^{-2 a}_b\tilde{A}^{-b} = 0\,. \lb{D--A}
\ee

Substituting $\tilde{A}^{-a}$ in the form \p{A-|and|E-2} into the
action \p{act-via-fA} and using \p{D--A}, we express
\p{act-via-fA} as
\be
\lb{act-via-fA-F}
S^{\rm bos} \ =\ \frac{i}{2}\int \, du dt  {\dot f}^{+a} \tilde{A}^{-b} \,   {\cal F}_{ab}   \, ,
\ee
where
\be
\lb{Fab}
{\cal F}_{ab}  \ = \ \left(\partial_{+[a} \partial_{-b]}
+ E^{-2c}_{[a} \partial_{-b]} \partial_{-c} \right) {\cal L}
\ee
is antisymmetric in  $a \leftrightarrow b$.

We are ready now to perform the third step and to plug in \p{act-via-fA-F} the
 solutions \p{f+sol} and \p{A-sol} derived above.
We obtain the Lagrangian
\be
\lb{L-harmint}
L^{\rm bos} \ =\ {\dot x}^{kb} {\dot x}^{lc} e^{\u{ja}}_{kb} e^{\u{ib}}_{lc}\, \int du \, u^+_{\u{j}} u^-_{\u{i}}
{\cal F}_{ab} (M^{-1})^a_{\u{a}}  (M^{-1})^b_{\u{b}} \, .
\ee
By using symmetry considerations, we can observe now that the whole expression  that multiplies
the structure $u^+_{\u{j}} u^-_{\u{i}}$ in \p{L-harmint} is antisymmetric
 over $\u{j} \leftrightarrow \u{i}$. Thus,
$u^+_{\u{j}} u^-_{\u{i}}$ may be replaced by $\varepsilon_{\u{ji}}/2$ and we finally derive
\be
L^{\rm bos} \ =\ \frac12\, g_{ia, jb}{\dot x}^{ia} {\dot x}^{jb}\,,
\ee
where
\be
\lb{metric}
g_{ia, jb} \ =\  e^{\u{ia}}_{ia} e^{\u{jb}}_{jb}\, \varepsilon_{\u{ij}} G_{\u{ab}} =
\varepsilon_{ij}\Omega_{ab} + {\rm nonlinear \,terms} \, ,
\ee
and
\be
\lb{Gab}
G_{\u{ab}} \ =\ \int du   \,  {\cal F}_{ab} (M^{-1})^a_{\u{a}}
(M^{-1})^b_{\u{b}}  =  \Omega_{\u{ab}} + {\rm nonlinear \,terms}\, .
\ee
We see that, generically, the metric is not just a convolution
of the matrices $e^{\u{ia}}_{ia}$, but involve an extra factor ---
the antisymmetric matrix $G_{\u{ab}}$ carrying the symplectic indices. However,
one can bring $G_{\u{ab}}$ to
the form $\Omega_{\u{ab}}$ using the   gauge freedom \p{M->RM} with a harmonic independent matrix $R$.
Indeed, consider $R_{\u{a}}^{\ \u{b}} = \delta_{\u{a}}^{\u{b}} +
\lambda_{\u{a}}^{\ \u{b}}$ with $\lambda \ll 1$. We can then write
\be
R_{\u{a}}^{\ \u{c}}  R_{\u{b}}^{\ \u{d}} \Omega_{\u{cd}} \ =\
 \Omega_{\u{ab}} - 2 \lambda_{\u{[ab]}} + o(\lambda) \, .
\ee
It is clear that a finite version of this transformation
can produce an arbitrary antisymmetric matrix $G_{\u{ab}}$.

Note that the metric \p{metric} as a whole is invariant
under the gauge transformations with matrix  $R_{\u{a}}^{\ \u{b}}$ defined
in \p{M->RM}. It is immediately seen from the definition \p{E+propu+}:
 a transformation of $M_a^{\u{b}}$ brings about the same
transformation of $e^{\u{jb}}_{jb}$, while the factors $\sim M^{-1}$
 in \p{Gab} are multiplied by the inverse matrices
$(R^{-1})_{\u{b}}^{\ \u{a}}$. If the gauge   $G_{\u{ab}} = \Omega_{\u{ab}}$ is chosen,
the metric is given by the standard expression
\be
\lb{metric-stand}
g_{ia, jb} \ =\  e^{\u{ia}}_{ia} e^{\u{jb}}_{jb}\, \varepsilon_{\u{ij}} \Omega_{\u{ab}} \, .
\ee
It is simply the familiar identity $g_{MN} = e_M^A e_N^A$ re-expressed in spinorial notation.

Consider now a set of models characterized by the same potential ${\cal L}^{+3a}$ and hence the same
harmonic constraints  \p{D++most}, but different Lagrangians ${\cal L}$ in the actions \p{S-lin}.
Different ${\cal L}$'s result in different ${\cal F}_{ab}$'s in \p{Fab}, which affect $G_{\u{ab}}$.
As we have just seen, this modification
 may be compensated by the appropriate gauge rotations of the bridges and,
correspondingly, of the vielbeins. We thus obtain a family of models whose vielbeins are
interrelated by the transformations
\be
\lb{e->eR}
\tilde{e}^{\u{ja}}_{ja} \ =\ e^{\u{jb}}_{ja} \, R_{\u{b}}^{\ \u{a}} \, ,
\ee
where all the dependence on ${\cal L}$ is encoded in the matrix $R$.
To keep $e^A_M$ real, the  matrix $R_{\u{b}}^{\ \u{a}}$ should
satisfy pseudoreality conditions, like in \p{pseudoreal}.
 Then it involves $4n^2$ real parameters.
 $2n^2 + n$ of them correspond to the action
 of $Sp(n)$ group,
  which does not affect the metric. So we are left with
 $2n^2 -n$ ``physical'' parameters. For $n=1$, only one parameter is left,
  which corresponds to multiplying the metric by a conformal factor.

An important observation is that the complex structures
in all such models coincide. Indeed, looking at the expression \p{StrSpin},
we observe that it is invariant under any $R$-transformations
 because the upper index $\underline{c}$ in  \p{StrSpin}
is rotated by the matrix $R$, while the contracted lower-case
index $\underline{c}$ by the inverse matrix $R^{-1}$. So we have
\be
\lb{tildeI=I}
(\tilde{I}^p)_{lc}{}^{kb} \ =\   (I^p)_{lc}{}^{kb} \, .
\ee
In particular, in all the models expressed in terms of linear
 ({\bf 4}, {\bf 4}, {\bf 0}) multiplets with vanishing
${\cal L}^{+3a}$,
the complex structures keep their flat form \p{Ip}.

An important corollary of this observation is
\begin{thm}
The Obata curvature invariants for a family of the HKT metrics,
characterized by a particular nonlinear constraint \p{D++q=F}
but having different Lagrangians ${\cal L}$ in \p{S-lin},  coincide.
\end{thm}

\begin{proof} This statement follows from Theorem 2, which says that
there exists a frame
 where the components of the Obata connection are expressed via the complex
structures, and the fact
 that the latter do not depend on ${\cal L}$. Then this connection also cannot depend on ${\cal L}$  and the same
 is true  for the Obata curvature invariants.
\end{proof}

It follows, in particular, that all  models based on linear $({\bf 4, 4, 0})$
multiplets are Obata-flat.

\section{HKT $\to$ HK $\scriptstyle{\Box}$ Taub-NUT}
\setcounter{equation}0

\begin{thm}
Consider a limited class of models where the function ${\cal L}^{+3a}$ in \p{D++q=F}
represents a gradient,
\be
\lb{L3=dL4}
{\cal L}^{+3a} \ =\ \Omega^{ab} \frac {\partial {\cal L}^{+4}}{\partial q^{+b}}
\ee
and the function ${\cal L}$ in \p{S-lin} is quadratic,
${\cal L} \ =\   q^{+a} q^-_a$.
In this case,  the metric is hyper-K\"ahler.
\end{thm}
\begin{proof}
The point is that the bosonic metric following from the quadratic
\p{S-lin} with the
constraints \p{D++q=F}, \p{L3=dL4} {\it exactly coincides} with the
metric derived from the
equation of motion \p{eqmotL+4} for the model \p{Sq+N8}.
 The latter involves the ``large'' multiplet \p{q+N8},
 it has ${\cal N} = 8$ supersymmetry and gives rise to a hyper-K\"ahler metric,
as follows from the theorem
of \cite{AG-F}. Thus, to prove our theorem, we only have to show that the
 metrics of the
two models are, indeed, the same.

Consider first the ${\cal N} = 4$ model. The bosonic action has the form
 \p{act-via-fA-F}, where the factor
\p{Fab} acquires now a very simple form ${\cal F}_{ab} = \Omega_{ab}$.
We obtain
\be
\lb{S-bosharm}
S \ =\ \frac{i}{2}\, \int dt du\, {\dot f}^{+a} \tilde{A}^-_a \, ,
\ee
where $f^{+a}$ and $\tilde{A}^-_a$  satisfy the constraints
\be
\partial^{++} f^{+a} &=& \Omega^{ab} \frac {\partial
 {\cal L}^{+4}}{\partial f^{+b}}\,,  \lb{cons-f} \\
{\cal D}^{++} \tilde{A}^-_a &=& -2i {\dot f}^+_a \, . \lb{cons-A}
\ee

We consider now the ${\cal N} = 8$ model \p{Sq+N8} and observe that
its bosonic action \p{S-comp-Q+} has exactly the
same form as \p{S-bosharm}   (one should only replace $f^{+a} \to F^{+a}$)
and the fields $f^{+a}$,
$\tilde{A}^-_a$ satisfy there exactly the same constraints \p{eqmot-F},
\p{eqmot-A} as for the
${\cal N} = 8$ model. This means that the metrics in these two models
 coincide.
As the metric of the  ${\cal N} = 8$ model is hyper-K\"ahler,
the same is true for the ${\cal N} = 4$ model.

\end{proof}

The hyper-K\"ahler nature of the metric follows also from the
explicit expressions for the Bismut connection and its torsion
derived
in \cite{DI} (see also Appendix\,C). If the condition \p{L3=dL4} is
fulfilled and ${\cal L}$ is quadratic, the torsion
vanishes and the Bismut connection is reduced to the Levi-Civita one.

Note that Theorem 8 only gives {\it sufficient}
conditions for the metric to be hyper-K\"ahler, but not
the necessary ones. In particular, if ${\cal L}$ is not
 quadratic and \p{L3=dL4} is not fulfilled,
the metric can still be hyper-K\"ahler [see Eq.~\p{harm-diff}
below and discussion thereof]. On the other hand,
for the quadratic ${\cal L}$, the condition \p{L3=dL4} is  also {\it necessary}.
 Indeed, let ${\cal L} = q^{+a} q^-_a$
and suppose that  \p{L3=dL4} is not fulfilled, but the metric is still HK.
 We know, however, from
the results of \cite{OgIv} discussed above that {\it any} HK metric is derived
from some prepotential
${\cal K}^{+4}$ which enters the ${\cal N} =8$ Lagrangian \p{Sq+N8}.
As we have just seen, the same metric
can be derived from the ${\cal N} = 4$ model \p{S-lin} with
 the quadratic ${\cal L}$ and the constraint involving
     \be
\lb{K3=dK4}
{\cal K}^{+3a} \ =\ \Omega^{ab} \frac {\partial {\cal K}^{+4}}{\partial q^{+b}} \, .
     \ee
But it is not possible that one and the same metric follows from two different
prepotentials ${\cal L}^{+3a} \neq {\cal K}^{+3a}$. It is clear from the derivation
outlined in Sect. 5.2 that the bridges and the vielbeins depend on the prepotential
in an essential way.

\vspace{1mm}

To understand how the procedure described in the preceding section
is actually carried through,
consider a simplest nontrivial hyper-K\"ahler example with $a=1,2$
(so that the manifold is 4-dimensional) and
\be
\lb{L+4TN}
{\cal L}^{+4} \ =\ -\frac 12\, (q^{+1})^2 (q^{+2})^2 \, .
\ee
In this case, the constraints \p{cons-f}, \p{cons-A}
can be solved explicitly, and the metric thus obtained is
the Taub-NUT metric (see \cite{CMP} and \cite{harm},
Chapter 5). Let us do here this calculation.
We start with the equation \p{cons-f} for
$f^{+a}$. In our case, this boils down to
\be
\lb{eq-f+}
\partial^{++} f^{+1} \ =\ (f^{+1})^2 f^{+2}, \qquad  \partial^{++} f^{+2} \ =\ -f^{+1} ( f^{+2})^2 \, .
\ee
The equations \p{eq-f+} can be easily solved. The solution reads
\be
\lb{sol-f+}
f^{+1} \ =\ \exp\{J\} x^{+1} \, , \qquad  f^{+2} \ =\ \exp\{ -J\}  x^{+2} \, ,
\ee
where $x^{\pm a} = x^{ia} u_i^\pm $ with harmonic-independent $x^{ia}$ and
\be
\lb{J}
J \ =\ \frac 12\left(x^{+1} x^{-2} + x^{-1} x^{+2} \right).
\ee

The solution to the constraint \p{cons-A} [which is the same
as in \p{tildeA-}] for $\tilde{A}^{-a}$ was given in \p{A-sol}.
Note that in the hyper-K\"ahler case the equation \p{D++most} for the bridge acquires the form
\be
\lb{eq-most-HK}
\partial^{++}  M_a^{\u{b}} + \Omega^{bc} \frac {\partial^2 {\cal L}^{+4}}{\partial f^{+a} \partial f^{+c}}
M_b^{\u{b}} \ =\ 0 \, .
\ee
The factor $ \Omega^{bc} \, \partial^2 {\cal L}^{+4}/ \partial f^{+a} \partial f^{+c}$
 belongs to the algebra $sp(n)$
[cf. \p{sympl-alg}] and that means that there exist solutions to the
equation \p{eq-most-HK} that belong to the group
$Sp(n)$. Such solutions satisfy the condition
\be
\lb{M-sympl}
     (M^{-1})_{\u{b}}^{c} \ =\  \Omega_{\u{ba}} \,
 M_a^{\u{a}}  \, \Omega^{ac}\, .
\ee
Bearing this in mind and the fact that ${\cal F}_{ab} = \Omega_{ab}$, the matrix \p{Gab} is also reduced to
$$
G_{\u{ab}} \ = \ \Omega_{\u{ab}} \, .
$$

Even after we impose the condition for the bridge $M$ to belong
 to $Sp(n)$, there is a freedom associated with multiplication by
constant $Sp(n)$ matrices.  We choose a solution that is reduced to the
unit matrix in the limit when the
nonlinearity associated  with ${\cal L}^{+4}$ is switched off. It reads
\be
\lb{most}
M_a^{\u{b}} \ =\ \frac 1 {\sqrt{1 + x^1 x^2}}
\left( \begin{array}{cc} e^{-J} (1 - x^{-1} x^{+2} ) &  e^{-J} x^{-2} x^{+2}  \\ [6pt]
-e^J x^{-1} x^{+1}  & e^J (1 + x^{-2} x^{+1} ) \end{array} \right),
\ee
where $x^1x^2 =-x^2 x^1 = x^{j1}x^2_j$. The matrix \p{most} belongs to $SU(2)$.

The vielbeins can be found from the solutions \p{sol-f+}, \p{most} and
from the definition \p{E+propu+}. They are
\be
\lb{TN-vierbeins}
e^{\u{ka}}_{ia}   \ =\ \frac 1{\sqrt{1 + x^1 x^2}}
\left( \begin{array}{cc}
- \frac 12 x_i^{\u{1}} x^{\u{k2}} + \delta^{\u{k}}_i \left( 1 + \frac 12 x^1 x^2 \right) &
- \frac 12 x_i^{\u{1}} x^{\u{k1}} \\  [6pt]
\frac 12 x_i^{\u{2}} x^{\u{k2}} &  \frac 12 x_i^{\u{2}} x^{\u{k1}} + \delta^{\u{k}}_i \left( 1 + \frac 12 x^1 x^2 \right)
\end{array} \right).
\ee

And the metric  \p{metric-stand} is
\be
\lb{TN-metric}
g_{ia, jb}  =  \frac 1{1 + x^1 x^2}  \! \left(\!\!\! \begin{array}{cc}
x^2_i x^2_j \left(1 + \frac 12 x^1 x^2 \right) & \!
\varepsilon_{ij}(1 + x^1x^2)^2 + x^2_i x^1_j \left(1 + \frac 12 x^1 x^2\right) \\ [6pt]
- \varepsilon_{ij}(1 + x^1x^2)^2 + x^1_i x^2_j\left(1 + \frac 12 x^1 x^2 \right) \! &
x^1_i x^1_j\left(1 + \frac 12 x^1 x^2 \right)
\end{array} \!\!\!\right).
\ee

It is equivalent to the Taub-NUT metric written in a familiar form
\begin{equation}\label{ds-TN}
ds^2= V_{\mathrm{TN}}^{-1}\,  \left(d\Psi+ \vec{A}d\vec{X}\right)^2+ V_{\mathrm{TN}} \, d\vec{X}d\vec{X} \,,
\end{equation}
where
\begin{equation}\label{V-TN}
V_{\mathrm{TN}}(r) \ = \ \frac{1}{r} +\lambda
\end{equation}
($r = |\vec{X}|$) and
\be
\label{U-TN}
A^1 \ =\  \frac{X^2}{r\left(r+X^3 \right)}\, , \ \ \
A^2 \ = \  -\frac{X^1}{r\left(r+X^3 \right)}\, , \ \ \
A^3 \ =\ 0
\ee
is the vector potential of a magnetic monopole.
The correspondence is established by setting in \p{V-TN} $\lambda = 1$ and changing the variables as
\be
\lb{def-3X}
X^p \ = \  (\sigma^p)_{ij}\,x^{i1}x^{j2}\,,  \ \ \ \ \ \
\Psi \ =\  -i\,\ln\left( \frac{x^{11}}{x^{22}}\right).
\ee
The relations
\be
r \ =\   x^1 x^2, \ \ \ \ \ \ \ \ \ \ d\Psi + \vec{A}d\vec{X}=\frac{i}{r} \left(x^2 dx^1 + x^1 dx^2 \right)
\ee
hold.
Inserting \p{V-TN}, \p{U-TN}  and \p{def-3X} in \p{ds-TN} we obtain
 \be
ds^2&=&\frac{1+\frac12\, x^1 x^2}{1+ x^1 x^2} \,
\Big( x^2_i x^2_j\, dx^{i1}dx^{j1} + x^1_i x^1_j\, dx^{i2}dx^{j2} + 2 x^2_i x^1_j\, dx^{i1}dx^{j2}\Big)  \nonumber\\
&&   +\,2\left(1+  x^1 x^2\right)dx^{i1}dx^{2}_i \,,\label{ds-TN1}
 \ee
which coincides with \p{TN-metric}. Note that the Taub-NUT metric
 possesses the  isometry $SU(2)\times U(1)$, with $SU(2)$ acting
as rotations of the doublet indices $i, j$ and $U(1)$ realized as
$\delta x^1_i = i\gamma\, x^1_i\,,  \;\delta x^2_i = -i\gamma\, x^2_i$.
This isometry is a corollary of the explicit $SU(2)\times U(1)$ invariance
of the hyper-K\"ahler potential \p{L+4TN}.

Another example where the constraints can be explicitly solved and the
 metric explicitly found is the model with
\be
\lb{Eg-Han-L4}
{\cal L}^{+4} \ =\ \frac {(\xi^{jk} u_j^+ u_k^+)^2}{(q^{+a} u_a^-)^2}
\ee
with an arbitrary symmetric $\xi^{jk}$. After performing the program outlined
above and choosing the coordinates
$x^\mu$ in appropriate way, one arrives \cite{harm} (see also \cite{GIOT})
at the hyper-K\"ahler Eguchi-Hanson metric
\be
\lb{Eg-Han}
ds^2 \ =\ \frac {dr^2}{1 - (a/r)^4} + r^2 \left\{\sigma_1^2 + \sigma_2^2 +
\left[1 -  \left( \frac a r\right)^4 \right] \sigma_3^2 \right\} \, ,
\ee
where $r = \sqrt{x_\mu^2}$ and
$\sigma_p$ are not the Pauli matrices, but   the Maurer--Cartan forms
$\sigma_p = \eta_{p\mu\nu} x^\mu d x^\nu$  (and $\eta_{p\mu\nu}$ are `t Hooft's symbols).
 The metric \p{Eg-Han} involves
 an axial $U(1)$ symmetry, which is also seen in the prepotential \p{Eg-Han-L4}.
It also possesses a $SU(2)$ isometry,
 the explicit realization of which can be found in \cite{harm}.

\section{Discussion}
\setcounter{equation}0

We have proven that a pair of functions --- ${\cal L}^{+3a}(q^{+b}, u)$
entering the constraint \p{D++q=F}  and ${\cal L}(q^{+b}, u)$
entering the action \p{S-lin} --- describes a HKT geometry.  But one
should understand that this description has a high degree of redundancy:
 by simply performing a  variable change, one can arrive from a given pair
 $({\cal L}^{+3}, {\cal L})$ to a pair that looks completely different,
whereas geometry is, of course, the same. It is the same redundancy which is
 incorporated in the description of the Riemannian geometry by the metric tensor.

As a simple example, consider a flat metric described by the linear multiplets $q^{+a}$
 satisfying the constraint
$D^{++} q^{+a} = 0$ and a quadratic ${\cal L}$. Introduce new variables
\be
\lb{harm-diff}
q'^{+a} \ = \ q^{+a} + C (q^{+b} u^-_b)^2 u^{+a} \, .
\ee
It is evident that after such a change ${\cal L}$ is not quadratic any more and
a nontrivial ${\cal L}'^{+3a}$ appears.

The geometries are distinguished by the curvature invariants. For the manifold
of interest,
a particular convenient tool is the Obata curvature, which coincides with the
ordinary Riemann curvature
for hyper-K\"ahler manifolds, but differs from it in  a general HKT case.
The convenience of Obata connections
and Obata curvatures stems from the fact that a frame exists where the former
are explicitly expressed
via the components of the complex structures (see Theorem 2), which are the same
for the whole {\it family} of metrics characterized by a particular ${\cal L}^{+3 a}$ and different ${\cal L}$.
That means that the Obata curvature invariants are also the same for all members of
this family.

Any such family of metrics has a distinguished  representative with ${\cal L} \propto q^{+a} q^-_a$.
Its geometry is simpler that the geometry of the other family members.
In particular, one can prove the following simple theorem:

\begin{thm}
The torsion form for an HKT model with
${\cal L} \propto  \, q^{+a} q^-_a$ is closed, $dC = 0$.
\end{thm}
In other words, we are dealing in this case with the so-called {\it strong} HKT geometry
\cite{strongHKT-def,GPS}.

\begin{proof}
A generic ${\cal N} = 4$ component Lagrangian \p{L(N=1)comp} includes a 4-fermion term.
The latter vanishes iff the torsion form is closed. Thus, it is sufficient to prove
that the full component action of the HKT model with quadratic ${\cal L}$ does not
involve such term.

Let us substitute in   ${\cal L} \ = \ q^{+a} q^-_a$ the expansions
\p{Expq+} and \p{q-a}. We derive
\be
\lb{Lcomp-ferm}
L \ =\ -\frac{1}{8}\,\int d^4\theta \, du \,{\cal L} \ =\
\frac{i}{2}\, \int du \left[ {\dot f}^{+a}(t, u) A^-_a(t,u) + \chi^a(t,u)
\dot {\bar \chi}_a(t, \tilde{u})\right]
+  {\rm total\ derivative} \, .
\ee
This Lagrangian involves only bi-fermion terms
[note that they are  also present in the  term $\sim {\dot f}A$
including the field $A^-_a(t,u)$ that satisfies
the constraint \p{A-}], while 4-fermion terms are absent.
\end{proof}

The same statement was proved in \cite{DI}, proceeding from the explicit expression for the torsion.

When ${\cal L}$ is not quadratic, the torsion form is not closed in most cases.
 But in some cases it can happen to be closed. For example, in the linear
case ${\cal L}^{+3a} = 0$  the metric coincides for $n=1$ with the 4-dimensional
flat metric multiplied by a certain conformal factor $G(x)$. The torsion form is
\cite{CKT-OKT}
\be
C \ \sim \ \varepsilon_{MNPQ}\, \partial_Q G(x) \, dx^M \wedge dx^N \wedge dx^P \, .
\ee
If $\partial_M \partial_M  G(x) = 0$, this form is closed. The general
criterions of closedness of the torsion in terms of the
potentials ${\cal L}^{+3 a}$ and ${\cal L}$ were given in \cite{DI}.

The results outlined above are illustrated in Table ~\ref{tabl:tab1}.
\begin{table}[h!]
\caption{HK and HKT geometries.}
\label{tabl:tab1}
\begin{center}
\renewcommand{\arraystretch}{2}
\begin{tabular}{|c|c|c|}
\cline{2-3}
\multicolumn{1}{c|}{}
& ${\cal L} = q^{+a} q^-_a$ & ${\cal L} \neq q^{+a} q^-_a$ \\
\hline
\multirow{2}{*}{${\displaystyle {\cal L}^{+3a}  =
\Omega^{ab} \,\frac {\partial {\cal L}^{+4}}{\partial q^{+b}}}$} & HK  &   \\
&$C=0$& weak HKT  \\
\cline{1-2}
\multirow{2}{*}{${\displaystyle {\cal L}^{+3a}  \neq \Omega^{ab}
 \,\frac {\partial {\cal L}^{+4}}{\partial q^{+b}}}$} & strong HKT   &
$dC \neq 0$ (generically)  \\
&$dC = 0$&  \\
\hline
\end{tabular}\\
\end{center}
\end{table}

The vielbeins for the different HKT manifolds belonging to a
 particular  family described above are interrelated
by the transformations \p{e->eR}. As was mentioned, these
 transformations involve
$2n^2 - n$ relevant parameters. For $n=1$, there is only one
such parameter --- a conformal factor
by which a metric can be multiplied. The families including
a hyper-K\"ahler metric as a member (in the four-dimensional case, these are the manifolds
 conformally equivalent to hyper-K\"ahler ones)
play a distinguished role. Let us call the HKT metrics belonging to
such families {\it reducible}
and the metrics not related to any hyper-K\"ahler metric by a transformation
\p{e->eR} {\it irreducible}.
One can prove the following noteworthy theorem:

\begin{thm}
Irreducible HKT metrics exist.
\end{thm}
This is a nontrivial statement. Even though the ``stem member''
 of a given HKT family with quadratic
${\cal L}$ may have a nontrivial ${\cal L}^{+3a}$ not given by \p{L3=dL4}
and the corresponding manifold is not hyper-K\"ahler,
it is not obvious that the constraints cannot be brought to a form which satisfy
\p{L3=dL4}  by a variable change
$q^{+a} \to    q^{\prime +a}$ [after which ${\cal L}$ is not quadratic any more --- cf.
the discussion around \p{harm-diff}].
There are many such metrics which seem to be irreducible,
but are in fact reducible in disguise.

\begin{proof}
To prove the theorem, it is sufficient to indicate at
least one example of an irreducible metric.
Such a metric was constructed in Ref. \cite{DV}.
Consider a 4-dimensional model ($n=1$) with
${\cal L} = q^{+a}q^-_a$ and the constraints
\be
\lb{cons-DV}
D^{++} q^{+1} \ \equiv \ {\cal L}^{(+3)1} &=&   (q^{+1})^2 q^{+2} (\lambda + i\rho) \nn
D^{++} q^{+2} \ \equiv \ {\cal L}^{(+3)2}  &=& -  q^{+1} (q^{+2})^2  (\lambda - i\rho)
\ee
with  real $\lambda, \rho$. When $\rho = 0$, we go back to the Taub-NUT system, but if
$\rho \neq 0$, $\partial_a {\cal L}^{+3a} \neq 0$,  the constraints \p{cons-DV} are not
expressed as in \p{L3=dL4} and we are dealing with a nontrivial
HKT manifold. Using a general technique described in Sect. 5,
one can derive its metric \footnote{We quote Eq.~(24) of Ref.~\cite{DV} where we set $\gamma_0 = 1$.}:
\be
\lb{metr-DV}
4d\tau^2 \ =\ V^{-1}(s) (d\Psi + \omega)^2 + V(s) \Gamma \, ,
\ee
where
\be
V(s) \ =\ \frac 1s + \lambda, \qquad \omega \ =\ \cos
\theta d\phi - \frac {\rho(1 +\lambda s)}{1 + \rho^2 s^2}ds\, , \nn
\Gamma \ =\ \frac {ds^2}{(1 + \rho^2 s^2)^2} +
\frac {s^2}{1 + \rho^2 s^2} (d\theta^2 + \sin^2 \theta d\phi^2) \, .
\ee
In the limit $\rho \to 0$, this metric goes over to the Taub-NUT metric \p{ds-TN}.
For $\rho \neq 0$, the Bismut connection of this model involves a nontrivial torsion.
The explicit expression for the torsion form is \cite{DV}
\be
\lb{torsion-DV}
C \ =\ \frac {\rho s (1 + \lambda s)}{(1 + \rho^2 s^2)^2}\,
ds \wedge \sin \theta d\theta \wedge d\phi \, .
\ee
In accord with Theorem 9, this form is closed.

Now, if the HKT metric \p{metr-DV} were reducible,   a certain conformal transformation
\be
\lb{popytka}
d\tilde{\tau}^2 \ =\ G(s, \Psi, \theta, \phi) d\tau^2
\ee
making the metric $d\tilde{\tau}^2$  hyper-K\"ahler would exist.
Equivalently, the vielbeins would be conformal
to the hyper-K\"ahler ones, with a real conformal factor $\sqrt{G}$.
Then the ordinary Levy-Civita covariant derivatives
of the complex structures [the latter do not change under such conformal
transformations --- see Eq.\p{tildeI=I}]
with the Christoffel symbols following from  \p{popytka}
must vanish:
\be
\lb{nab-I=0}
\nabla_M \, I_N^{\ P}  =  \nabla_M \, J_N^{\ P}  =  \nabla_M \, K_N^{\ P}
\ =\ 0 \, .
\ee
This condition can be solved to determine the conformal factor $G$ \cite{Callan},
but it turns out
\footnote{We address the readers to the papers \cite{Valent,Papa} for details.}
that in the particular case of the Delduc-Valent metric \p{metr-DV}, this procedure gives
a complex function $G(s, \Psi, \theta, \phi)$ and a complex ``hyper-K\"ahler metric''.
Real solutions to the constraints \p{nab-I=0} do not exist.

\end{proof}

There is an interesting question that has not been completely clarified yet. In this paper,
we explained in relatively simple terms that the harmonic action \p{S-lin}
with the constraints \p{D++q=F} lead to a HKT geometry.
One can conjecture that {\it any} HKT metric may be described in this way, but we did
not prove that.

As we mentioned at the end of Sect. 3, we know that any hyper-K\"ahler metric {\it is}
determined by the harmonic prepotential ${\cal L}^{+4}$. This was derived without
resorting to supersymmetry,
but solving the constraint that the Riemann curvature form $R_A^{\ B}$ belongs to
$sp(n)$. A similar program was carried out for strong HKT manifolds (with closed torsion form)
in \cite{DeKaSo}. By
solving the constraint that the curvature form $\hat{R}_A^{\ B}$ of the {\it Bismut}
connection belongs to
$sp(n)$, the authors showed that any strong HKT geometry is described by the prepotential  ${\cal L}^{+3a}$
(up to a difference in notation). It would be interesting
to generalize this result to the weak HKT geometry and show that any HKT metric is
derived from the
data including ${\cal L}^{+3 a}$ and arbitrary (not necessarily quadratic) ${\cal L}$.

It would be also very interesting to build up a supersymmetric description of generic
bi-HKT manifolds, not only those
that were described by the linear multiplets. One may guess
  \cite{DI} that this would require extending the framework of ${\cal N}=4, d=1$
harmonic superspace to the bi-harmonic superspace \cite{IvNied}.

\bigskip
\section*{Acknowledgements}

We are indebted to Fran\c{c}ois Delduc for the collaboration at the initial stage of this project and
to Gueo Grantcharov for useful discussions.  S.F. \& E.I. acknowledge support from the RFBR Grant No. 18-02-01046,
Russian Science Foundation Grant No. 16-12-10306,
Russian Ministry of Education and Science, project No. 3.1386.201 and a grant
of the IN2P3-JINR Programme. They would like to thank SUBATECH, Universit\'{e} de Nantes,
for the warm hospitality in the course of this study.

\section*{Appendix A: One-dimensional harmonic superspace}
\def\theequation{A.\arabic{equation}}
\setcounter{equation}0

We give here only the basic formulas that we use in the main text. \footnote{We mostly follow the conventions of  Ref. \cite{IvLe},
where more details can be found, but  the sign of
 $t$ and $t_A$ in Eqs. \p{tA}, \p{D++} and \p{D--} below is opposite  compared to that in
 Refs. \cite{IvLe,DI} and coincides with the convention used in  Ref.  \cite{Maxim1}
  and the review \cite{FIL12}. The latter convention is more convenient, giving a more natural
sign for the   fermion kinetic term and more natural quantization prescription \cite{FIL09}.}
Consider the case ${\cal N} = 4$. We trade the ordinary
${\cal N}{=}\,4$, $d{=}\,1$ supercoordinates $\left(t,\theta_i, \bar\theta^i \right)$, $\bar\theta^i=(\theta_i)^\ast$, $i=1,2$,
by the supercoordinates
\begin{equation}\label{harm-scoord}
\left(t\,,\theta^\pm\,, \bar\theta^\pm\, , u^\pm_i \right)\,,
\end{equation}
where additional commuting variables $u^\pm_i$ are the $SU(2)/U(1)$ harmonics:
\begin{equation}\label{harm-1}
u^+{}^i u^-_i=1\,,\qquad (u^+{}^i)^\ast=u^-_i\,,
\end{equation}
and new Grassmann coordinates are the harmonic projections of $\theta_i, \bar \theta^i$,
\begin{equation}\label{harm-theta}
\theta^\pm=\theta^i u^\pm_i\,,\qquad \bar\theta^\pm=\bar\theta^i u^\pm_i\,.
\end{equation}
It is convenient to define a generalized ``tilde'' conjugation $A\to\widetilde A$ which is a combination
of usual complex conjugation and antipodal reflection on the sphere $SU(2)/U(1)$.
This transformation acts as
\begin{equation}\label{harm-tilde}
\widetilde{\theta^\pm}=\bar\theta^\pm\,,\qquad \widetilde{\bar\theta^\pm}=-\theta^\pm\,,\qquad
\widetilde{u^\pm_i}=u^\pm{}^i\,.
\end{equation}
We introduce then the analytic time\, \footnote{In most formulas written in the main text, the index $A$ is not displayed, however.}
\be
\lb{tA}
t_A = t + i(\theta^+ \bar \theta^- + \theta^- \bar \theta^+)\,,
\ee
which is invariant under \p{harm-tilde}.

We need also the harmonic covariant derivatives defined according to
\be
\lb{D0}
D^0 &=&  \partial^0 +   \theta^+ \frac {\partial}{\partial \theta^+ }
+  \bar\theta^+ \frac {\partial}{\partial \bar\theta^+} -  \theta^- \frac {\partial}{\partial \theta^- } -
\bar \theta^- \frac {\partial}{\partial \bar \theta^- }\,, \\
D^{++} &=& \partial^{++}  + 2i\theta^+ \bar \theta^+ \frac {\partial}{\partial t_A} +
\theta^+ \frac {\partial}{\partial \theta^- } +  \bar \theta^+ \frac {\partial}{\partial \bar\theta^- } \,,
\lb{D++} \\
D^{--} &=& \partial^{--} + 2i\theta^- \bar \theta^- \frac {\partial}{\partial t_A} +
\theta^- \frac {\partial}{\partial \theta^+ } +  \bar \theta^- \frac {\partial}{\partial \bar\theta^+ } \,,
\lb{D--}
\ee
where
\be
\lb{part-pmpm}
\partial^0 \ =\ u_i^+ \frac {\partial}{\partial u_i^+ } -  u_i^- \frac {\partial}{\partial u_i^- }\, , \qquad
\partial^{++} \ =\ u^+_i \frac {\partial}{\partial u_i^-} \, , \qquad
\partial^{--} \ =\ u^-_i \frac {\partial}{\partial u_i^+} \, .
\ee

All the superfields that we deal with have a definite {\it harmonic charge}
$D^0$. For example, $q^{+a}$ carries harmonic charge 1, ${\cal L}^{+4}$
carries harmonic charge 4, etc.
The harmonic derivatives satisfy the algebra
\be
\lb{harm-DEP-alg}
[D^{++}, D^{--}] \ =\ D^0 \, , \qquad [D^0, D^{\pm \pm} ] \ =\ \pm D^{\pm \pm} \, ,
\ee
the same as for the ``short'' derivatives \p{part-pmpm}.

The basic convenience of the harmonic approach stems from the fact that one can formulate the theory in terms of
{\it analytic superfields} living in analytic superspace
\begin{equation}\label{harm-scoord-AS}
\left(t_A\,,\theta^+\,, \bar\theta^+\, , u^\pm_i \right) \ \equiv\ (\zeta, u) \, .
\end{equation}
Correspondingly, their $\theta$ expansion is much shorter than for generic superfields that depend in addition on $\theta^-$ and $\bar \theta^-$.

Note that, though the basic superfields that we use are analytic, the Grassmann measure in Eq.~\p{S-lin} and the subsequent formulas
refers to the full superspace,
\be
\lb{Gr-measure}
d^4\theta \ =\ d\theta^+ d\bar \theta^+ d\theta^- d\bar \theta^- \,, \quad \int d^4\theta (\theta^+\bar\theta^+\theta^-\bar\theta^-) =1\,.
\ee
The harmonic integrals are normalized to $\int du = 1$.

We define ${\cal N}{=}\,8$ harmonic superspace (that we use it in Sect. 3) in a similar way.  The odd coordinates $\theta_{i\alpha}$
carry now the extra
index $\alpha = 1,2$ and $\bar \theta^{i\alpha} \stackrel{\rm def}= (\theta_{i\alpha})^*$. We consider the harmonic projections
of these coordinates
with respect to the index $i$,
$\theta^{\pm}_{\alpha}=\theta^{i}_{\alpha}u^{\pm}_i$ .
In analogy to  \p{harm-scoord-AS} we define the analytic harmonic superspace involving four Grassmann coordinates
$\theta^{+}_{\alpha}$, $\bar\theta^{+\alpha}$. The expressions for the analytic time $t_A$  and harmonic covariant derivatives
have the same form as \p{tA} and
\p{D0}--\p{D--}, but  involve an extra summation over $\alpha$. The Grassmann integration measure  over the analytic superspace
is defined as in \cite{harm},
\be
\int  d^2 \theta^+ d^2 \bar \theta^+\, (\theta^+)^2 (\bar \theta^+)^2 =1\,, \quad
(\theta^+)^2 := \theta^+_\alpha\bar \theta^{+ \alpha}\,,
(\bar\theta^+)^2 := \bar\theta^{+ \alpha}\bar \theta^{+}_{\alpha}\,.
\ee

\setcounter{equation}{0}

\section*{Appendix B: Canonical form of complex structures}
\def\theequation{B.\arabic{equation}}
\setcounter{equation}0

Here we show how to reduce  generic tangent space quaternionic complex structures $(I^p)_{AB}$ to the constant matrices \p{I-HK}, using
the tangent space gauge freedom $O(4n)$ and the defining quaternionic conditions \p{quatern}.

We will use the spinor notation, in which the relation \p{quatern} amounts to the following one
\be
(I^p)_{ai}^{\;bj} (I^q)_{bj}^{\;ck} = -\delta^{pq}\delta^c_a\delta^k_i  + \varepsilon^{pqu}(I^u)_{ai}^{\;ck}\,.\lb{QAspin}
\ee
Due to the antisymmetry property, $(I^p)_{ai\;bj} = - (I^p)_{bj\;ai},$ we can parametrize  $(I^p)_{ai\;bj}$ as
\be
(I^p)_{ai\;bj} = \varepsilon_{ij} {\cal F}^p_{(ab)} + {\cal B}^p_{[ab]\,(ij)}\,,\lb{FandB}
\ee
where, for the moment, ${\cal F}^p$ and ${\cal B}^p$ are arbitrary functions. The complex structures are transformed
under the general tangent space $O(4n)$
gauge rotations as
\be
\delta(I^p)_{ai\;bj} = \Lambda_{ai}^{\;a'i'} (I^p)_{a'i'\; bj} + \Lambda_{bj}^{\;b'j'} (I^p)_{ai \;b'j'}\,.
\ee
Here,
\be
\Lambda_{ai}^{\;a'i'} = \varepsilon^{i'k}\Omega^{a' b}\Lambda_{ai\, bk}\,, \quad \Lambda_{ai\, bk} = -\Lambda_{bk\, ai}
:= \varepsilon_{ik} \Lambda_{(ab)} + \Lambda_{[ab]\,(ik)}
\ee
and $\Lambda_{(ab)}, \Lambda_{[ab]\,(ik)}$ are arbitrary gauge parameters, representing, respectively,
$Sp(n) \subset O(4n)$ [$n(2n + 1)$ parameters] and the coset
$O(4n)/Sp(n)$ [$3n(2n-1)$ parameters]. For the functions defined in \p{FandB}, these infinitesimal transformations
amount to the following ones
\be
&& \delta {\cal F}^p_{(ab)} = \big[\Lambda_{(ad)}\Omega^{dc}{\cal F}^p_{(cb)} + (a \leftrightarrow b)\big]
-\frac12 \big[ \Lambda_{[ad]\,(il)}\Omega^{dc} {{\cal B}^p_{[cb]}}{\,}^{(il)} + (a \leftrightarrow b)\big], \nn
&& \delta {\cal B}^p_{[ab]\,(ij)} = \big[\Lambda_{(ad)} \Omega^{dc} {\cal B}^p_{[cb]\,(ij)}
+ \Lambda_{[ad]\,(ij)}\Omega^{dc}{\cal F}^p_{(cb)} -(a \leftrightarrow b)\big]  \nn
&& \qquad\qquad - \,\frac12\big[\Lambda_{[ad]\,(i}^{\qquad l)}\Omega^{dc} {\cal B}^p_{[cb]\;(lj)}
+\Lambda_{[ad]\,(j}^{\qquad l)}\Omega^{dc} {\cal B}^p_{[cb]\;(li)} -
(a \leftrightarrow b)\big].
\ee

The natural assumption is that $(I^p)_{ai\;bj}$ have the flat limit in which $(I^p)_{ai\;bj} \,\Rightarrow\, i \Omega_{ab} (\sigma^p)_{ij}$.
Extracting this constant part from ${\cal B}^p_{[ab]\,(ij)}$,
\be
{\cal B}^p_{[ab]\,(ij)} \quad \Rightarrow \quad i \Omega_{ab} (\sigma^p)_{ij} + \tilde{{\cal B}}^p_{[ab]\,(ij)},
\ee
and looking at the inhomogeneous part of the transformation of $\tilde{{\cal B}}^p_{[ab]\,(ij)}\,$,
\be
\delta\tilde{{\cal B}}^p_{[ab]\,(kj)} = i \big[ \Lambda_{[ab]\, (kl)} (\sigma^p)^l_j + (k \leftrightarrow j)\big] + \ldots\,,\lb{InHom}
\ee
we observe that the $O(4n)/Sp(2n)$ gauge transformations are capable to gauge away
some components of ${\cal B}^p_{[ab]\,(kj)}\,$.
Expanding
$$
{\cal B}^p_{[ab]\,(kj)} = {\cal B}^{p q}_{[ab]}\,(\sigma^q)_{(kj)}\,,
$$
it can be shown that one can, e.g., gauge away ${\cal B}^{1 2}_{[ab]}, {\cal B}^{1 3}_{[ab]}$ and ${\cal B}^{3 1}_{[ab]}$,
after which the
$O(4n)/Sp(n)$ gauge freedom gets fully exhausted and we end up with the subgroup $Sp(n)$ as the only residual
tangent space gauge group.\footnote{It is worth to note that one
cannot choose the gauge like ${\cal B}^{1q}_{[ab]} = 0$ or, e.g., ${\cal B}^{2q}_{[ab]} = 0$;
using the gauge transformations \p{InHom}, one can remove two components from
${\cal B}^{p_0q}_{[ab]}$, $p_0$ being fixed, and one component from ${\cal B}^{pq}_{[ab]}\,, \;p \neq p_0\,$.}

After this total use of the tangent space gauge freedom, there still remain ${\cal F}^p_{(ab)}, {\cal B}^{1 1}_{[ab]},
{\cal B}^{3 2}_{[ab]}, {\cal B}^{3 3}_{[ab]}$ and the whole set
${\cal B}^{2 q}_{[ab]}$. Now we are going to show that all of them can be eliminated by exploiting the quaternionic algebra \p{QAspin}.

We start from
\be
&&(I^1)_{ai\;bj} = \varepsilon_{ij} {\cal F}^1_{[ab]} + \big[i\Omega_{ab} + {\tilde{\cal B}}^{11}_{[ab]}\big] (\sigma^1)_{ij}\,, \lb{I1} \\
&&(I^3)_{ai\;bj} = \varepsilon_{ij} {\cal F}^3_{[ab]} + i\Omega_{ab}(\sigma^3)_{ij} + {\tilde{\cal B}}^{31}_{[ab]}(\sigma^1)_{ij}
+ {\tilde{\cal B}}^{32}_{[ab]}(\sigma^2)_{ij}\,. \lb{I3}
\ee
The relation \p{QAspin} with $p=q = 1$ implies
\be
&& 2i {\cal F}^1_{(ac)} - {\cal F}^1_{(a}{\,}^{\;b)}{\tilde{\cal B}}^{11}_{[b c]}
- {\cal F}^1_{(c}{\,}^{\;b)}{\tilde{\cal B}}^{11}_{[b a]} = 0\,, \lb{F1} \\
&& {\cal F}^1_{(a}{\,}^{\;b)}{\cal F}^1_{(bc)} - 2i{\tilde{\cal B}}^{11}_{[a c]}
+ {\tilde{\cal B}}^{11}_{[a}{\,}^{\;b]}{\tilde{\cal B}}^{11}_{[b c]} = 0\,. \lb{B11}
\ee
Eq. \p{F1} gives ${\cal F}^1_{(ac)} = 0$, then \p{B11} yields ${\tilde{\cal B}}^{11}_{[b c]} = 0$, whence
\be
(I^1)_{ai\;bj} = i\Omega_{ab}(\sigma^1)_{ij}\,. \lb{I1fin}
\ee

Next, making use of the gauge-fixed form \p{I3} for $I^3$,  we exploit \p{QAspin} with $p= 1, q=3\,$.
This relation yields ${\cal F}^2_{(ac)} = {\cal F}^3_{(ac)} =
{\tilde{\cal B}}^{31}_{[a c]} = {\tilde{\cal B}}^{21}_{[a c]} = {\tilde{\cal B}}^{22}_{[a c]} = 0\,,
{\tilde{\cal B}}^{23}_{[a c]} = -{\tilde{\cal B}}^{32}_{[a c]}$, thus
fixing $I^3$ and $I^2$ as
\be
(I^2)_{ai\;bj} = i\Omega_{ab}(\sigma^2)_{ij} + {\tilde{\cal B}}^{23}_{[a b]}(\sigma^3)_{ij}\,, \quad
(I^3)_{ai\;bj} = i\Omega_{ab}(\sigma^3)_{ij} - {\tilde{\cal B}}^{23}_{[a b]}(\sigma^2)_{ij}\,. \lb{I2I3fin}
\ee
Now it is straightforward to check that the complex structures \p{I1fin} and \p{I2I3fin}
satisfy the whole set of relations \p{QAspin} under the single condition
\be
{\tilde{\cal B}}^{23}_{[a d]} \Omega^{db}{\tilde{\cal B}}^{23}_{[b c]} = 0\,. \lb{Condfin}
\ee
Obviously, ${\tilde{\cal B}}^{23}_{[a d]} = 0$ is a solution of \p{Condfin}, and it yields
just the canonical constant form \p{Ip} for the tangent space complex structures. Actually ${\tilde{\cal B}}^{23}_{[a d]} = 0$
is the only solution which ensure the complex structures to form a triplet with respect to the global $SU(2)$ acting on the doublet
indices $i, k$.

The above proof is quite analogous to the statement that the complex structure for the $2n$ dimensional K\"ahler manifold
in the tangent-space representation can be reduced to the canonical constant form with non-zero holomorphic and anti-holomorphic
entries by fully fixing the gauge tangent space freedom $O(2n)/U(n)$ (see, e.g., \cite{harm}).  The specificity
of the hyper-K\"ahler case is that,
besides fixing the gauge freedom $O(4n)/Sp(n)$, one should essentially use the quaternionic algebra \p{QAspin} in order to reduce
the corresponding triplet of the complex structures to the canonical form   \p{Ip}.
\setcounter{equation}{0}

\section*{Appendix C: More on the Obata connection}
\def\theequation{C.\arabic{equation}}
\setcounter{equation}0

Here we derive the form of Obata connection in the manifestly $SU(2)$ covariant framework of Ref. \cite{DI}.
Note that this issue, like the one in Appendix D, were never discussed before.
We will use some results of \cite{DI} without explicit derivation. Also we stick to the conventions of the present
paper concerning the sign of vielbeins.

We will start with giving the tangent-space form of the Bismut connection
\be
\hat{\Gamma}_{\u{i} \u{a}, \u{k} \u{b}\, \u{l} \u{c}} = {\Gamma}_{\u{i} \u{a}, \u{k} \u{b} \,\u{l} \u{c}} +
\frac12 C_{\u{i} \u{a}\, \u{k} \u{b}\, \u{l} \u{c}}\,,  \quad  C_{\u{i} \u{a} \,\u{k} \u{b}\, \u{l} \u{c}} =
\varepsilon_{\u{i}\u{l}}\,\nabla_{\u{k}\u{a}}\, G_{[\u{c} \u{b}]} +
\varepsilon_{\u{k}\u{l}}\,\nabla_{\u{i}\u{b}} \,G_{[\u{a} \u{c}]}\,, \lb{BiCon}
\ee
where
\be
\nabla_{\u{i}\u{a}}\, G_{[\u{c} \u{b}]}  = \int du \,(M^{-1})^{a}_{\u{a}}(M^{-1})^{c}_{\u{c}}(M^{-1})^{b}_{\u{b}}\,
\big(\partial_{- a}{\cal F}_{cb}\,u^-_{\u{i}} + {\cal D}_{+ a} {\cal F}_{cb}\,u^+_{\u{i}}\big) \lb{DefnablaG}
\ee
and
\be
&& {\cal D}_{+ a} {\cal F}_{cb} = \nabla_{+ a}{\cal F}_{cb} + E^{-d}_{ac}{\cal F}_{db} + E^{-d}_{ab}{\cal F}_{cd}\,, \quad
 \nabla_{+ a} = \partial_{+a} + E^{-2 d}_a \partial_{-b}\,, \nn
&& {\cal D}^{++}  E^{-d}_{ac} =  E^{+d}_{ac}\,, \quad E^{+d}_{ab} = \partial_{+a}\partial_{+b}\,{\cal L}^{+ 3 d}\,. \lb{DefE-daac}
\ee
The derivations $\nabla_{\u{i}\u{a}}\, G_{[\u{c} \u{b}]} $ satisfy the cyclic identity
\be
\nabla_{\u{i}\u{a}}\, G_{[\u{b} \u{c}]}  + {\rm cycle} (\u{a}, \,\u{c}, \, \u{b}) = 0\,, \lb{Cyclic}
\ee
from which it is easy, e.g.,  to prove the total antisymmetry  of $C_{\u{i} \u{a}\, \u{k} \u{b}\, \u{l} \u{c}}$ in \p{BiCon} with
respect to the permutations of the tangent space index pairs.

Let us now deform  $C_{\u{i} \u{a}\, \u{k} \u{b}\, \u{l} \u{c}}$  in the following way
\be
C_{\u{i} \u{a}\, \u{k} \u{b}\, \u{l} \u{c}} \; \Rightarrow \; \tilde{C}_{\u{i} \u{a}\, \u{k} \u{b}\, \u{l} \u{c}} :=
C_{\u{i} \u{a}\, \u{k} \u{b}\, \u{l} \u{c}}  - 2 \varepsilon_{\u{i} \u{l} }\, \nabla_{\u{k}\u{c}}\, G_{[\u{a} \u{b}]}\,.\lb{ObataA}
\ee
Using the identity \p{Cyclic}, it is straightforward to check that
\be
\tilde{C}_{\u{i} \u{a}\, \u{k} \u{b}\, \u{l} \u{c}} = \tilde{C}_{\u{i} \u{a}\, \u{l} \u{c}\, \u{k} \u{b}}\,,
\ee
so this tensor defines a new symmetric connection
\be
\tilde{\Gamma}_{\u{i} \u{a}, \u{k} \u{b}\, \u{l} \u{c}} = {\Gamma}_{\u{i} \u{a}, \u{k} \u{b} \,\u{l} \u{c}} +
\frac12 \tilde{C}_{\u{i} \u{a}, \u{k} \u{b}\, \u{l} \u{c}}\,.  \lb{ObataB}
\ee
which proves to be just the Obata connection as the quaternionic complex structures turn out to be covariantly constant
with respect to it. Note that the tensor $\tilde{C}$ is defined up to the addition
\be
\sim \varepsilon_{\u{k} \u{l} }\, \nabla_{\u{i}\u{a}}\, G_{[\u{b} \u{c}]}\,, \lb{extraCond}
\ee
which is symmetric with respect to the permutation $\u{k}\u{b}\, \Leftrightarrow \,\u{l}\u{c}$ in itself. However,
such a structure is ruled out by the requirement of the covariant constancy of the quaternionic complex structures
with respect to $\tilde{\Gamma}$ \p{ObataB}.

Let us show that the quaternionic complex structures are indeed covariantly constant with respect to \p{ObataB}.
It will be convenient to write the triplet of constant complex structures in the tangent space representation as
\be
I_{(\u{k}  \u{l})}{}^{\;\;\u{i} \u{a}}_{\u{j} \u{b}} = \frac{i}2 \delta^{\u{a}}_{\u{b}}
\big(\varepsilon_{\u{k} \u{j}}\,\delta^{\u{i}}_{\u{l}}  + \varepsilon_{\u{l} \u{j}}\,\delta^{\u{i}}_{\u{k}}\big)\,.\lb{klCompl}
\ee
The condition of the covariant constancy of this complex structure with respect to any affine connection $\hat{\tilde{\Gamma}}$
in the tangent space representation amounts to the following general form of $\hat{\tilde{\Gamma}}$ \footnote{The second term in
\p{GenSol} is none other than the properly restricted spin connection.}
\be
 \hat{\tilde{\Gamma}}^{\u{i} \u{a}}_{\u{k} \u{b}\, \u{l} \u{c}}
 = e_{\u{k} \u{b}}^{{k} {b}}\partial_{kb} e^{\u{i} \u{a}}_{lc} e^{lc}_{\u{l} \u{c}}
+ \delta^{\u{i}}_{\u{l}} \,{\cal F}^{\u{a}}_{\u{k} \u{b} \u{c}}\,.\lb{GenSol}
\ee
In particular, for Bismut connection
\be
{\cal F}^{\u{a}}_{\u{k} \u{b} \u{c}} = G^{[\u{a}\, \u{d}]}\nabla_{\u{k}\u{c}}\, G_{[\u{d}\, \u{b}]} - D^{\;\;\u{a}}_{\u{k} \;\u{b}\, \u{c}}\,,
\ee
where
\be
D^{\;\;\u{a}}_{\u{k} \;\u{b}\, \u{c}} = e^{lb}_{\u{k} \,\u{b}}\,\partial_{lb} M^{\u{a}}_{d} (M^{-1})^d_{\u{c}} +
u^+_{\u{k}} \,E^{- \u{a}}_{\u{b}\, \u{c}}\,.\lb{DefinD}
\ee
Taking into account that
\be
 \tilde{\Gamma}^{\u{i} \u{a}}_{\u{k} \u{b}\, \u{l} \u{c}} = \varepsilon^{\u{i}\u{j}}\,G^{[\u{a}\,\u{d}]}\,
\tilde{\Gamma}_{\u{j} \u{d},  \u{k} \u{b}\, \u{l} \u{c}} = \hat{\Gamma}^{\u{i} \u{a}}_{\u{k} \u{b}\, \u{l} \u{c}} -
\delta^{\u{i}}_{\u{l}}\,G^{[\u{a}\,\u{d}]}\,\nabla_{\u{k} \u{c}}\,G_{[\u{d}\,\u{b}]}\,,
 \ee
 we see that the complex structures  \p{klCompl}, as well as their world-index  images,
 are covariantly constant with respect to the connection $\tilde{\Gamma}$
 \be
 \tilde{\cal D}_{ia} I_{(\u{k} \u{l})}{\,}^{\;\;jb}_{tc} =0\,.
 \ee
 So the symmetric connection just defined obeys the full set of the properties required for the Obata connection and can be identified with
 the latter\footnote{The structure \p{extraCond} does not match with the general form \p{GenSol} and so breaks the covariant constancy
 condition.}. By {\bf Lemma 4}, it is unique and so should reduce to \p{Obata} in the special frame lacking manifest $SU(2)$ covariance.
 It is also instructive to represent this
 connection in the general form \p{GenSol}
 \be
 \tilde{\Gamma}^{\u{i} \u{a}}_{\u{k} \u{b}\, \u{l} \u{c}} = e_{\u{k} \u{b}}^{{k} {b}}\partial_{kb} e^{\u{i} \u{a}}_{lc} e^{lc}_{\u{l} \u{c}}
 - \delta^{\u{i}}_{\u{l}} \,D^{\;\;\u{a}}_{\u{k} \;\u{b}\, \u{c}}\,.
 \ee

 As was already mentioned, the metric   is not covariantly constant
 with respect to the Obata connection, as opposed to the Levi-Civita or Bismut connections. Indeed, it is easy to find
 \be
  e_{\u{i} \u{a}}^{{i} {a}}  e_{\u{k} \u{b}}^{{k} {b}}  e_{\u{l} \u{c}}^{{l} {c}} \big(\tilde{\cal D}_{ia}\,g_{kb,\,lc}\big) =
  \varepsilon_{\u{k} \u{l}}\,\nabla_{\u{i} \u{a}}\,G_{[\u{b}\,\u{c}]} \neq 0\,.\lb{noncovar}
 \ee
 Note that the hyper-K\"ahler  case corresponds to the condition $\nabla_{\u{i} \u{a}}\,G_{[\u{b}\,\u{c}]} =0$ \cite{DI}
 \footnote{In the hyper-K\"ahler case, ${\cal F}_{ab} = \Omega_{ab}\,, \; E^{-d}_{ac} = \Omega^{db}E^{-}_{(bac)}\,, \;
 E^{+d}_{ac} = \Omega^{db}E^{+}_{(bac)}$ and $E^{+}_{(bac)} = \partial_{+b}\partial_{+a}\partial_{+c}\,{\cal L}^{+ 4}$, which implies
 $\partial_{-a}{\cal F}_{bc} = {\cal D}_{+a}{\cal F}_{bc} = 0$ and so the vanishing of the r.h.s. of \p{DefnablaG} .}
 under which the torsion
 in \p{BiCon}  and the r.h.s. of  \p{noncovar} vanish and the Obata connection gets identical
 to the Levi-Civita connection.
\setcounter{equation}{0}

\section*{Appendix D: ${\cal N}=4$ transformations of fermions}
\def\theequation{D.\arabic{equation}}
\setcounter{equation}0

Our last topic is the derivation of the transformation law of the central basis harmonic-independent
fermionic fields under ${\cal N}=4$ supersymmetry.

The  ${\cal N}=4$ supersymmetry transformation laws of the harmonic-dependent  fermionic fields $\chi^a$ and $\bar{\chi}^a$ are
as follows \cite{DI}
\be
&& \delta \chi^a = 2i\bar{\epsilon}^- \dot{f}^{+ a} - \bar{\epsilon}^+ \big[ 2i \big(\dot{f}^{-a} - E^{-2a}_b \dot{f}^{+ b}\big)
- E^{-a}_{bc}\bar{\chi}^b\chi^c\big], \nonumber \\
&&  \delta \bar\chi^a =  - 2i\epsilon^- \dot{f}^{+ a} + {\epsilon}^+ \big[ 2i \big(\dot{f}^{-a} - E^{-2a}_b \dot{f}^{+ b}\big)
- E^{-a}_{bc}\bar{\chi}^b\chi^c\big]. \lb{chi1}
\ee

After some algebra, for the harmonic-independent fermionic variables $\psi^{\u{i}\u{a}} = (\psi^{\u{a}}, \bar\psi^{\u{a}})$ defined in
\p{combining-psi} we  obtain the following transformation rule
\be
\delta \psi_{\u{i}}^{\;\;\u{a}} = -2i \,\epsilon_{\u{i}\u{j}}\,\dot{x}^{kd}\, e^{\u{j}\u{a}}_{kd} -
\epsilon_{\u{k}\u{j}}\, D^{\u{j}\u{a}}_{\;\;\u{b}\u{c}}\,\psi^{\u{k}\u{b}}\,\psi_{\u{i}}^{\;\;\u{c}}\,, \lb{Prom}
\ee
where $\epsilon_{\u{i}\u{j}} = \varepsilon_{\u{j}\u{k}}\epsilon^{\u{k}}_{\u{i}}, \;\epsilon^{\u{k}}_{\u{1}}
\equiv \epsilon^{\u{k}}, \,\epsilon^{\u{k}}_{\u{2}}
\equiv \bar\epsilon^{\u{k}}\,.$ Next, we pass to $\psi^{ia} = e^{ia}_{\u{k}\u{b}}\,\psi^{\u{k}\u{b}}$ and find the ${\cal N}=4$
transformation of  $\psi^{ia}$ as
\be
\delta \psi^{ia} = \delta x^{kd} \,\partial_{kd}e^{ia}_{\u{k}\u{b}}\,\psi^{\u{k}\u{b}} + e^{ia}_{\u{k}\u{b}}\,\delta \psi^{\u{k}\u{b}}\,,
\ee
where $\delta x^{kd}$ is given by \p{delx=epsi-matr}.
At this step, it is convenient to split the transformation parameter into the singlet and traceless triplet parts as in \p{eps-kl},
$\epsilon^{\u{k}}_{\u{i}} = i\epsilon_0 \delta^{\u{k}}_{\u{i}} + \epsilon^p (\sigma^p)^{\u{k}}_{\u{i}}$. After some work, we obtain for
the singlet transformation
\be
\delta_0 \psi^{ia} =  -2\epsilon_0 \, \dot{x}^{ia} - i\epsilon_0\, e^{ia}_{\u{l}\u{b}}\big( \partial_{[jc} e^{\u{l}\u{b}}_{kd]} -
e^{\u{t}\u{a}}_{[jc}e^{\u{l}\u{g}}_{kd]}\,D^{\;\;\u{b}}_{\u{t}\;\u{a}\u{g}}\big)\psi^{jc}\psi^{kd}.
\ee
The expression within the brackets vanishes due to the identity (5.32) in Ref. \cite{DI} \footnote{This identity and some
other identities listed in \cite{DI} can be derived
from the definitions of vielbeins, $\partial_{kb}f^{+ a} M_{a}^{\u{c}} = e^{\u{k}\u{c}}_{kb} u^+_{\u{k}}$,
$(\partial_{kb}f^{- a} - E^{-2a}_c\partial_{kb}f^{+ c}) M_{a}^{\u{d}}  = e^{\u{k}\u{d}}_{kb} u^-_{\u{k}}\,,$ and the definition \p{DefinD}.}.
Then
\be
\delta_0 \psi^{ia} =  -2\epsilon_0 \, \dot{x}^{ia}\,,\lb{N1ferm}
\ee
thus confirming the property that  $\psi^{ia}$ and $x^{ia}$ for each index pair $ia$ are components of ${\cal N}=1$
multiplet\footnote{To achieve
the literal  correspondence with the ${\cal N}=1$ transformation rule \p{epsilon0comp} one is led to rescale the time variable in
\p{N1ferm} as $ t \rightarrow  2t$.}.

After some efforts, the triplet part of the variation can be cast in the following nice form
\be
\delta_3\psi^{ia} = 2 \epsilon^{p} (I^{p})^{\;\;\;ia}_{kd}\dot{x}^{kd} -i \epsilon^{p} (I^p)_{[jd}^{\;\;kc}\,\tilde{\Gamma}^{ia}_{tb] \,kc}\,
\psi^{jd}\psi^{tb}\,. \lb{Tripl}
\ee
Here, $\tilde{\Gamma}^{ia}_{tb \,kc}$ is Obata connection discussed in Appendix C. We see that this geometric object essentially enters
the  nonlinear ${\cal N}=4$ transformation laws of the harmonic-independent fermionic fields in the general ${\cal N}=4$ supersymmetric
HKT model. Presumably, it is not difficult to check that in the special coordinate frames the transformation law \p{Tripl}
for each value of $p$  coincides with \p{delXN2-comp} which follows from the ${\cal N}=1$ superfield formalism and lacks
manifest $SU(2)$ symmetry. We leave it for  an inquisitive reader.

\end{document}